\documentclass[11pt]{article} 
\usepackage[utf8]{inputenc} 

\usepackage[margin=1in]{geometry} 
\usepackage{graphicx} 

\usepackage{booktabs} 
\usepackage{array} 
\usepackage{paralist} 
\usepackage{verbatim} 
\usepackage{mathrsfs}
\usepackage{amssymb}
\usepackage{amsthm}
\usepackage{amsmath,amsfonts,amssymb}
\usepackage{esint}
\usepackage{graphics}
\usepackage{enumerate}
\usepackage{mathtools}
\usepackage{xfrac}
\usepackage{subcaption}
\usepackage{stmaryrd}

\let\oldsquare\square 

\usepackage{mathabx}

\renewcommand{\square}{\oldsquare}

\usepackage[usenames,dvipsnames]{xcolor}
\usepackage[colorlinks=true, pdfstartview=FitV, linkcolor=blue, citecolor=blue, urlcolor=blue]{hyperref}
\usepackage[normalem]{ulem}

\usepackage{tikz}
\usetikzlibrary{calc}
\usepackage{pgf}
\usetikzlibrary{external}
\numberwithin{equation}{section}
\numberwithin{figure}{section}

\newtheorem{theorem}{Theorem}[section]

\newtheorem{corollary}[theorem]{Corollary}
\newtheorem{proposition}[theorem]{Proposition}
\newtheorem{lemma}[theorem]{Lemma}
\theoremstyle{definition}
\newtheorem{definition}[theorem]{Definition}

\newcommand*{\N}{\ensuremath{\mathbb{N}}}

\newcommand*{\Z}{\ensuremath{\mathbb{Z}}}

\newcommand*{\R}{\ensuremath{\mathbb{R}}}
\newcommand*{\Zd}{\ensuremath{\mathbb{Z}^d}}
\newcommand*{\Rd}{\ensuremath{\mathbb{R}^d}}

\newcommand{\eps}{\varepsilon}

\renewcommand*{\tilde}{\widetilde}

\renewcommand{\P}{\ensuremath{\mathbb{P}}}

\newcommand{\f}{\mathbf{f}}

\newcommand{\A}{\mathcal{A}}

\renewcommand{\S}{\mathcal{S}}

\DeclareMathOperator{\dist}{dist}

\DeclareMathOperator{\var}{var}

\newcommand{\E}{\mathbb{E}}
\DeclareMathOperator{\diam}{diam}

\DeclareSymbolFont{boldoperators}{OT1}{cmr}{bx}{n}
\SetSymbolFont{boldoperators}{bold}{OT1}{cmr}{bx}{n}
\usepackage{accents}
\newcommand\thickbar[1]{\accentset{\rule{.45em}{.6pt}}{#1}}
\renewcommand{\bar}{\thickbar}
\renewcommand{\a}{\mathbf{a}}

\newcommand{\ahom}{\bar{\a}}
\newcommand{\fhom}{\bar{\f}}

\renewcommand{\S}{\mathcal{S}}
\newcommand{\T}{\mathcal{T}}
\newcommand{\Sp}{\S_P}

\def\Xint#1{\mathchoice
{\XXint\displaystyle\textstyle{#1}}%
{\XXint\textstyle\scriptstyle{#1}}%
{\XXint\scriptstyle\scriptscriptstyle{#1}}%
{\XXint\scriptscriptstyle\scriptscriptstyle{#1}}%
\!\int}
\def\XXint#1#2#3{{\setbox0=\hbox{$#1{#2#3}{\int}$}
\vcenter{\hbox{$#2#3$}}\kern-.5\wd0}}

\def\fint{\Xint-}

\makeatletter
\newcommand\avsuminner[2]{%
  {\sbox0{$\m@th#1\sum$}%
   \vphantom{\usebox0}%
   \ooalign{%
     \hidewidth
     \smash{\,\rule[.23em]{8.8pt}{1.1pt} \relax}%
     \hidewidth\cr
     $\m@th#1\sum$\cr
   }%
  }%
}
\makeatother

\makeatletter
\newcommand\avsuminnerr[2]{%
  {\sbox0{$\m@th#1\sum$}%
   \vphantom{\usebox0}%
   \ooalign{%
     \hidewidth
     \smash{\,\rule[.23em]{6pt}{0.7pt} \relax}%
     \hidewidth\cr
     $\m@th#1\sum$\cr
   }%
  }%
}
\makeatother

\let\originalleft\left
\let\originalright\right
\renewcommand{\left}{\mathopen{}\mathclose\bgroup\originalleft}
\renewcommand{\right}{\aftergroup\egroup\originalright}

\newcommand{\cu}{\square}

\newcommand{\indc}{{\boldsymbol{1}}}

\renewcommand{\hat}{\widehat}

\usepackage{titlesec}

\newcommand{\addperiod}[1]{#1.}
\titleformat{\section}
   {\centering\normalfont\Large}{\thesection.}{0.5em}{}
\titleformat*{\subsection}{\bfseries}
\titleformat{\subsubsection}[runin]
  {\normalfont\bfseries}
  {\thesubsubsection.}
  {0.5em}
  {\addperiod}
\titleformat*{\subsubsection}{\normalfont\itshape}
\titleformat*{\paragraph}{\bfseries}
\titleformat*{\subparagraph}{\large\bfseries}

\title{The scaling limit of the continuous solid-on-solid model}

\author{Scott Armstrong
\thanks{Courant Institute of Mathematical Sciences, New York University.
{\footnotesize \href{mailto:scotta@cims.nyu.edu}{scotta@cims.nyu.edu}.}
}
\and
Wei Wu
\thanks{Mathematics Department, NYU Shanghai \& NYU-ECNU Institute of Mathematical Sciences.
{\footnotesize \href{mailto:wei.wu@nyu.edu}{wei.wu@nyu.edu}.}
}
}

\date{October 20, 2023 }

\usepackage[nottoc,notlot,notlof]{tocbibind}

\begin{document}

\maketitle

\begin{abstract}
We prove that the scaling limit of the continuous solid-on-solid model in~$\Zd$ is a multiple of the Gaussian free field. 

\end{abstract}


\section{Introduction}

We consider the real-valued random field~$\phi:\Zd\to \R$ with law given formally by the Gibbs measure 
\begin{equation}
\label{e.mu}
d\mu :=\frac 1Z  \exp\Biggl( - \sum_{x,y\in\Zd, \, |x-y|=1} |\phi(x) -\phi(y)|  \Biggr) \,d\phi\,.
\end{equation}
Here $d\phi$ denotes Lebesgue measure on $\R^{\Zd}$ and $Z$ is a normalizing constant, called the \emph{partition function}, which makes~$\mu$ a probability measure. The measure~$\mu$ is rigorously defined as the weak limit as~$L\to \infty$ of finite-volume Gibbs states~$\mu_L$ defined on the space of configurations restricted to the cube~$Q_L:=[-L,L]^d\cap \Zd$ with zero boundary data.  It is known (see Sheffield~\cite{She} or Section~\ref{s.tight} below) that these finite-volume measures~$\mu_{L}$ converge weakly, as~$L\to \infty$, to a unique measure~$\mu$ on the space of gradient fields (or, equivalently, fields modulo additive constants). This limit is our infinite-volume Gibbs state, and the random field it describes is called the \emph{continuous solid-on-solid model.}

\smallskip 

The main result of this paper is the following theorem asserting that the continuous solid-on-solid model~\eqref{e.mu} in $\Zd$ converges to a Gaussian free field in the scaling limit.  

\begin{theorem}[Scaling limit of the continuous SOS model]
\label{t.CLT}
There exists a constant~$\ahom >0$ such that, for every~$\f=(\f_1,\ldots,\f_d)\in C_c^\infty(\Rd;\Rd)$, the sequence of random variables 
\begin{equation}
\label{e.FR}
F_R(\nabla \phi) := R^{-\frac d2}\sum_{x\in\Zd}\sum_{i=1}^d \f_i\left(\frac xR\right) \left( \phi(x+e_i)-\phi(x) \right), \quad R\geq 1
\end{equation}
converges in law, as~$R\to \infty$, to a normal random variable with zero mean and variance equal to
\begin{equation}
\label{e.quantttty}
\int_{\Rd} \nabla \cdot \f (x) \cdot G_{\ahom}(x,y)\nabla \cdot \f (y)\,dx\,dy, 
\end{equation}
where
\begin{equation*}
G_{\ahom}(x):= 
\left\{ 
\begin{aligned}
& -(2\pi \ahom)^{-1} \log |x| & \mbox{if} & \ d=2\,, \\
& (d(d-2)|B_1|\ahom)^{-1}  |x|^{2-d} & \mbox{if} & \ d>2\,,
\end{aligned}
\right.
\end{equation*}
is the Green function for the elliptic operator~$\ahom\Delta$ in~$\Rd$.
\end{theorem}

The continuous solid-on-solid model is an analogue of its \emph{integer-valued} counterpart, which arises in the study of interfaces separating two pure phases in the three-dimensional Ising model. Indeed, if we consider the three-dimensional Ising model and send the interaction strength along one of the axes to infinity, then the interface becomes a random function. The law of this random function is the integer-valued solid-on-solid model~\cite{vBe,BFL},  which has the same form as~\eqref{e.mu} with the added constraint that the field~$\phi$ takes values in~$\beta \Z$, for some~$\beta>0$.

\smallskip

The Gibbs measure described in~\eqref{e.mu} is a particular instance in the general family of~$\nabla\phi$-interface models, which are defined by
\begin{equation}
\label{e.V}
d\mu :=\frac 1Z  \exp\left( - \sum_{x,y\in\Zd, \, |x-y|=1} \mathsf V(\phi(x) -\phi(y)) \right) \,d\phi,
\end{equation}
where the \emph{interaction potential} $\mathsf{V}:\R \to \R$ is typically assumed to be an even function with some minimal growth at infinity. In the case that~$\mathsf{V}$ is quadratic, this model coincides with the (a constant times the discrete) Gaussian free field. It is believed that, under very general assumptions which include the case~$\mathsf{V}(x)=|x|$ of interest here, the limiting large-scale statistical behavior of the field~$\nabla \phi$ is a Gaussian free field~\cite{BLL,She,Funaki}.

\smallskip

A lot of effort has been made in analyzing this model in the under the assumption that~$\mathsf{V}$ is uniformly convex and has a second derivative which is uniformly bounded from above. In this case, Naddaf and Spencer~\cite{NS} were the first to characterize the scaling limit for~$\nabla \phi$ as a Gaussian free field with a positive definite covariance matrix~$\ahom$. Their arguments were based on the observation that the \emph{Helffer-Sj\"ostrand representation}~\cite{HS} reveals an equivalence between the scaling limit and a corresponding elliptic homogenization problem. The assumptions on~$\mathsf{V}$ correspond to a uniform ellipticity assumption, and the matrix~$\ahom$ is the homogenized diffusion matrix. 
Other works in the uniformly convex case include the characterization of the hydrodynamic limit~\cite{FS} for the corresponding Langevin dynamics associated to the Gibbs measure, a characterization of the fluctuations~\cite{GOS} of these dynamics,  central limit theorems for finite volume measures~\cite{Mi}, quantitative limit theorems and the~$C^2$ regularity of the surface tension~\cite{AW,AD2}, and a local limit theorem~\cite{Wu}.

\smallskip

The scaling limits of the $\nabla \phi$ interface models with non-uniformly convex potentials have been obtained in several settings,  albeit not covering the solid-on-solid model~\eqref{e.mu}. Using renormalization group arguments, Adams, Buchholz, Koteck\'y and M\"uller~\cite{AdKM,AdBKM} studied certain (possibly non-convex) potentials which are~$C^k$ perturbations of uniformly convex ones which are uniformly convex in a large ball centered at the origin but may be nonconvex far from the origin. They characterized the scaling limit to a Gaussian free field and proved the strict convexity of the surface tension near the origin (see also~\cite{CDM} for an earlier work in this direction). Biskup and Spohn~\cite{BiSp} obtained the scaling limit for a very special class of non-convex potentials that can be written a mixture of Gaussians. Biskup and Rodriguez~\cite{BR} obtained an invariance principle for random walks in a class of degenerate, dynamical random environment, inspired by the $\nabla\phi$ model with potential 
\begin{equation}
\label{e.soft}
\mathsf V(x) 
= 
\left\{
\begin{aligned}
 &\frac12 x^2\,, & \mbox{if}  \ |x|\leq 1\,,
\\ 
& |x|-\frac 12\,, & \mbox{if}  \ |x|>1\,,\end{aligned}
\right.
\end{equation}
which behaves like the potential~$\mathsf V(x) = |x|$ for large~$x$ but is uniformly convex near the origin sets.  

\smallskip

Constraining the model to take values in a  discrete set introduces another type of nonconvexity which makes the analysis much more difficult. Obtaining the scaling limit for the discrete Gaussian model at high temperature--- which is an integer valued height function with the quadratic potentials---was recently obtained by Bauerschmidt, Park and Rodriguez~\cite{Bau1,Bau2}, using renormalization group methods. 
This model is believed to undergo a so-called roughening transition when $d=2$:  for large $\beta$,  the height function is believed to be localized; for small $\beta$ the height function is believed to be rough, and converges to a Gaussian free field scaling limit.  
For the (integer-valued) solid-on-solid model, a quantitative logarithmic lower bound of the variance in the rough phase was first derived by Frohlich and Spencer~\cite{FrSp}, using a connection to the Coulomb gas and an elaborate multiscale expansion. We also mention the work by Lammers and Ott~\cite{LO} (see also~\cite{Lam}), who gave a simplified proof of the delocalization of the height function in more general settings (but without quantifying the logarithmic variance). 
The scaling limit of the integer-valued SOS model to a Gaussian free field remains an important open question. 

\smallskip

The main difficulty in the case of the continuous solid-on-solid model, and the reason the theory for it is very sparse (and in particular its scaling limit has remained open), is due to the singularity of~$\mathsf{V}(x) =|x|$ at the origin coupled with the complete lack of strict convexity. In particular, the singularity of the second derivative of~$\mathsf{V}$ near the origin makes it difficult to make sense of the Helffer-Sj\"ostrand representation. 

\smallskip

There have been works obtaining upper and lower bounds for the variance of the random variables in~\eqref{e.FR} under general assumptions which include the continuous solid-on-solid model. 
Using a modified Mermin-Wagner type argument, Milos and Peled~\cite{MiP} obtained a logarithmic lower bound for the~$\nabla \phi$ model in very generic setting (including the hard-core potential). For a logarithmic upper bound in two dimensions, see the work of Magazinov and Peled~\cite{MP} which covers some convex potentials with~$\{x: \mathsf V''(x) =0\}$ having measure zero, as well as the work of Dario~\cite{Da} which allows for potentials having asymptotic growth condition $\mathsf V''(x) \sim |x|^{p-2}$ for $p\in(2,\infty)$. The earlier work of Sheffield~\cite{She} proves the strict convexity of the surface tension for very general models which also include the continuous solid-on-solid model. 

\smallskip

Closer to the approach of this paper is the work by Brydges and Spencer~\cite{BS},  who studied the~$\nabla \phi$-model with the potential~$\mathsf V(x) = \sqrt{\delta + (\nabla \phi)^2}$, for~$\delta>0$. They obtain, among other results, a Brascamp-Lieb inequality~\cite{BL} and bounds on the exponential moments of the field. 
Their arguments are based on the integral identity~\eqref{e.magic} below and an integration by parts which transforms questions about the correlation structure of the random field into questions about the behavior of the elliptic Green function of a highly degenerate elliptic (discrete) PDE with random coefficients (see~Section~\ref{s.BS} below). 
A homogenization limit for this PDE to a multiple of the Laplacian would consequently yield the scaling limit of the field to a multiple of the GFF. Obtaining such a homogenization result is however difficult, due to the highly degenerate nature of the PDE and a lack of quantitative ergodicity: its random coefficients are unbounded from above and below and we do not possess good decorrelation estimates for it beyond qualitative ergodicity. 

\smallskip

Our proof of Theorem~\ref{t.CLT} starts from the insight of~\cite{BS}, that the moments of the continuous solid-on-solid field~\eqref{e.mu} are related to the degenerate elliptic homogenization problem, and combines it with \emph{coarse-graining methods} in homogenization (see~\cite{AKMbook,AKbook} and the references therein). We emphasize that, while the connection to the elliptic homogenization problem in~\cite{BS} carries over to the~$\delta =0$ case by a limiting argument, the law of the random coefficient field behaves very differently, and is much more degenerate, in the~$\delta=0$ case. Indeed, if~$\delta>0$, then all positive and negative moments of the conductance are bounded, and therefore one may apply the regularity theory for elliptic equations with moments conditions to obtain the homogenization.  These applies to the models treated in~\cite{BS} and~\cite{Ye}. For the solid-on-solid model case ($\delta=0$), the conductances only have finite \emph{inverse} moments---the field itself does not even have a finite first moment.
As a replacement for uniform ellipticity upper bounds, we prove subcritical percolation-type estimates on the level sets of the coefficient field, allowing us to deduce that the statistics of large conductances are dominated by subcritical percolation clusters (see Section~\ref{s.perco}). 

To make use of such a weak ``ellipticity'' condition, we adapt the recent coarse-graining methods developed in the context of quantitative homogenization in~\cite{AS,AKMbook,AKbook}. These ideas are quite flexible and allow such extremely degenerate coefficient fields, provided that they \emph{become (almost) uniformly elliptic after renormalization}, in other words, after coarse-graining on some very large scale. This is exactly what our subcritical percolation estimates are able to show.

\smallskip

The coarse-graining method has been previously employed in the context of a very degenerate equation in~\cite{AD1}, which proved quantitative homogenization results for harmonic functions on supercritical percolation clusters (arbitrarily close to criticality). In that setting, it is the uniform ellipticity lower bound that is found only after coarse-graining, while in the present paper it is the upper bound. However, the general idea of \emph{renormalizing a degenerate equation until it becomes uniformly elliptic} remains the same. 

\smallskip

It is notable that, while the coarse-graining methods were developed for \emph{quantitative} homogenization, they are robust enough for our use here---despite a lack of any quantitative ergodicity condition on the coefficient field. The key is that we are able to quantify the scale at which the coarse-grained fields are uniformly elliptic, even if we cannot  quantify the (larger) scale at which homogenization occurs. This is because the former is controlled by the percolation estimates (which we can quantify), while the latter is due to the ergodic properties of the field (which we cannot). 

\smallskip 

This paper is organized as follows.  In the next section we  reduce Theorem~\ref{t.CLT} to a degenerate elliptic homogenization problem by adapting the ideas of Brydges and Spencer~\cite{BS}.  In Section~\ref{s.perco} we obtain percolation estimates for the conductances, which implies the probability that the origin belongs to a large cluster of high or low conductances is very small.   In Section~\ref{s.homog} we introduce the subadditive energy quantities and show by a multiscale iterative argument that they converge to the limiting dual quantities. We then use a variational approach to prove Theorem~\ref{t.CLT}.  Some auxiliary estimates are stated in Appendix.

\section{Reduction to elliptic homogenization problem}
\label{s.BS}

In this section we reduce the main result to an elliptic homogenization problem by adapting the ideas of Brydges and Spencer~\cite{BS}.

\subsection{Derivation for finite volume Gibbs measures}

Throughout, we denote~$Q_L:= [-L,L]^d \cap \Zd$.  We let~$\mathcal{E}(Q_L)$ denote the set of directed edges on~$Q_L$ and $Q_L^\circ$  the \emph{interior} of $Q_L$.   We introduce a one parameter family of finite volume Gibbs measures that generalize~\eqref{e.mu}: 
\begin{equation}
\label{e.mudelta}
d\mu_{L,\delta} := \frac1{Z_{L,\delta}} \exp\Biggl( - \sum_{x\in Q_L^\circ}\sum_{y \sim x} (\delta + |\phi(x) -\phi(y)|^2)^{\sfrac12}  \Biggr) \,d\phi\,,
\end{equation}
where $Z_{L,\delta}$ is a normalizing constant that makes $\mu_{L,\delta}$ a probability measure.  The derivation in this section and the main result of this paper works for all $\delta\ge 0$.   

\smallskip

As in~\cite{BS}, the starting point of our derivation is the integral identity
\begin{equation}
\label{e.magic}
\exp \bigl( -2\sqrt{z}\bigr) 
= 
\pi^{-\sfrac12}
\int_{-\infty}^\infty \exp \biggl( -z \exp(t)   -\exp (-t) -\frac12t \biggr)  \,dt\,.
\end{equation}
Let us quickly recall the derivation of~\eqref{e.magic} from~\cite{BS}, which reduces the integral to the integral of a Gaussian after a change of variables:
\begin{align*}
\lefteqn{
\int_{-\infty}^\infty \exp \biggl( -z \exp(t)   -\exp (-t) -\frac12t \biggr)  \,dt
} \qquad &
\notag \\ & 
=
z^{\sfrac14} \int_{-\infty}^\infty \exp \Bigl ( -z^{\sfrac12} \exp(s)   -z^{\sfrac12}\exp (-s) -\frac12s \Bigr )  \,ds
\qquad \mbox{(substitution: $s = t +\frac12\log z$)}
\notag \\ &
=
\frac12 z^{\sfrac14} \int_{-\infty}^\infty 
\Bigl ( \exp\bigl ( \tfrac12s \bigr ) + \exp\bigl ( -\tfrac12s \bigr )   \Bigr )
\exp \Bigl ( -z^{\sfrac12} \exp(s)   -z^{\sfrac12}\exp (-s)  \Bigr )  \,ds
\notag \\ &
=
\frac12 z^{\sfrac14} \int_{-\infty}^\infty 
\Bigl ( \exp\bigl ( \tfrac12s \bigr ) + \exp\bigl ( -\tfrac12s \bigr )   \Bigr )
\exp \Bigl ( -z^{\sfrac12} \bigl ( \exp\bigl ( \tfrac12s \bigr ) -\exp\bigl ( -\tfrac12s \bigr ) \bigr )^2  -2z^{\sfrac12}\Bigr )  \,ds
\notag \\ &
=
\int_{-\infty}^\infty 
\exp \Bigl ( -u^2 - 2z^{\sfrac12}\Bigr )  \,du
\qquad \mbox{(substitution: $u = z^{\sfrac14}\exp\bigl ( \tfrac12s \bigr ) -\exp\bigl ( -\tfrac12s \bigr )$)}
\notag \\ &
=
\pi^{\sfrac12}\exp \bigl ( - 2z^{\sfrac12}\bigr ) \,.
\end{align*}
Applying~\eqref{e.magic}, we may write the partition function as 
\begin{equation}
\label{e.jointZ}
Z_{L,\delta} = \int e^{-A_L(\phi,\tau,\delta)} \prod_{e \in \mathcal E(Q_L)} \exp(-e^{-\tau_e} - \frac12 \tau_e) \,d\tau d\phi\,,
\end{equation}
where 
\begin{equation}
A_L(\phi,\tau,\delta) = \sum_{e \in \mathcal E(Q_L)}  (\delta + (\nabla\phi(e))^2) e^{\tau_e}: = (\phi, D_L(\tau) \phi) + \delta \sum_{e \in \mathcal E(Q_L)}  e^{\tau_e}
\,.
\end{equation}
Here we introduced the quadratic form: for every $f: Q_L \to \R$,
\begin{equation}
\label{e.dl}
(f, D_L(\tau) f):= \sum_{e \in \mathcal E(Q_L)}   (\nabla f(e))^2 e^{\tau_e}
\,.
\end{equation}
We may therefore view the measure~\eqref{e.jointZ} as a joint measure on $(\phi,\tau)$. Condition on any realization of $\tau$, the $\phi$ marginal is Gaussian, thus the conditional expectation can be computed exactly and we obtain the $\tau$ marginal of the measure

\begin{equation}
\label{e.Zt}
Z_{L,\delta,\tau} = \int \frac{1}{\sqrt{\det D_L(\tau)}} \prod_{e \in \mathcal E(Q_L)}   \exp\Bigl(-\delta e^{\tau_e}-e^{-\tau_e} - \frac12 \tau_e\Bigr) \,d\tau_e 
\end{equation}
and
\begin{equation}
\label{e.muLtau}
d\mu_{L,\tau}^\delta := \frac{1}{Z_{L,\delta,\tau} } \frac{1}{\sqrt{\det D_L(\tau)}} \prod_{e \in \mathcal E(Q_L)}   \exp\Bigl(-\delta e^{\tau_e}-e^{-\tau_e} - \frac12 \tau_e\Bigr) \,d\tau_e \,.
\end{equation}
The derivation above also applies to evaluate the moments of linear functionals of $\nabla\phi$. Denote by $\f_R := \f(\frac{\cdot}{R})$, and take $L>R$. By applying~\eqref{e.magic} and  integrating out the $\phi$ marginal, we obtain
\begin{align}
\var_{\mu_L}[F_R] 
&
= \var_{\mu_L} \Bigl[R^{-\frac d2}\sum_{x\in\Zd}\sum_{i=1}^d f_i\left(\frac xR\right) \left( \phi(x+e_i)-\phi(x) \right) \Bigr]
\notag \\ &
= 
R^{-d} \sum_{x,y\in \Zd}\left\langle \nabla \cdot \f_R (x) G_\a(x,y)  \nabla \cdot \f_R (y) \right\rangle_{\mu_{L,\tau}^\delta}\,.
\end{align}
Here $G_\a =(\nabla \cdot \a \nabla)^{-1}$,  where $\a$ is a diagonal matrix with $\a(e,e) = e^{\tau_e}$. Let $u$ be the solution of the PDE
\begin{equation}
\nabla\cdot  \a \nabla u = \nabla\cdot \f_R
\quad \mbox{in} \ \Zd\,.
\end{equation}
with Dirichlet boundary condition $u=0$ on $\partial Q_L$.  Also, given $f,g: \Zd \to\R$,  define the inner product
\begin{equation}
\label{e.inner}
(f,g)_R:= R^{-d} \sum_{x\in\Zd} f(x)g(x).
\end{equation}
We may write the above variance as 
\begin{equation}
\var_{\mu_L}[F_R]   =\left\langle (\f_R, \nabla u)_R \right\rangle_{\mu_{L,\tau}^\delta}\,.
\end{equation}

\subsection{Derivation in infinite volume}

The preceding argument does not directly apply to the infinite-volume Gibbs measure $\mu$, since the integral in~\eqref{e.jointZ} diverges in  infinite volume, and we may not direct apply the transformation in~\eqref{e.magic}. 

Instead, we derive the elliptic PDE based on the infinite-volume Gibbs state~\eqref{e.mu} constructed by Sheffield~\cite{She}.  Let $\mu_\delta$ be the infinite-volume Gibbs state  defined by 
\begin{equation}
\label{e.mudelta.inf}
d\mu_{\delta} := \frac 1{Z_{\delta}} \times \exp\left( - \sum_{x\in \Zd}\sum_{y \sim x} (\delta + |\phi(x) -\phi(y)|^2)^{\sfrac12}  \right) \,d\phi\,,
\end{equation}
where $Z_{\delta}$ is the normalization constant to make $\mu_\delta$ a probability measure. 
By~\cite{She}, $\mu_\delta$ is the unique translation invariant ergodic Gibbs state on $\Zd$ with zero slope.

We denote by $\Omega$ the set of functions $\phi: \Zd \to \R$. Following the notation of~\cite{BiSp}, for any finite collection of edges $E \subset \mathcal E(\Zd)$, define a coupling of the $\phi$ and $\tau$ field  $\tilde \mu_E^\delta$ on $\Omega \times \R^E$ by 

\begin{equation}
\label{e.couple}
\tilde \mu_E^\delta (A \times B) :=
\int_B \prod_{e\in E}  e^{-\delta e^{\tau_e} - e^{- \tau_e} - \frac {\tau_e}{2}}
\biggl  \langle \indc_A \prod_{e\in E} e^{\sqrt{\delta+ (\nabla \phi(e))^2} - (\nabla \phi(e))^2e^{\tau_e} } \biggr  \rangle_{\!\!\mu_\delta} \, d\tau,
\end{equation}
for every $A\subset \Omega,  B\subset  \R^E$. We may think of $\tilde \mu_E^\delta $ as obtained by applying the formula~\eqref{e.magic} only to the edge set $E$ in the definition of $\mu_\delta$ in~\eqref{e.mudelta.inf}.  It is straightforward to check (using~\eqref{e.magic}) that $(\tilde \mu_E^\delta)_{E \subset \mathcal E(\Zd)}$ is a consistent family of measures. By the Kolmogorov extension theorem, there exists a unique measure $\tilde \mu^\delta$ on $\Omega \times \R^{\mathcal E(\Zd)}$ such that $\tilde \mu_E^\delta $ is its projection on $\Omega \times \R^E$. 
By studying the corresponding $\tau$-measure in finite volumes~\eqref{e.muLtau}, we prove in Section~\ref{s.tight} that the sequence of the measures $\{ \mu_{L,\delta}\}_{L\ge 1}$ is tight. 

\begin{lemma}
\label{l.tight}
The sequence of the $\nabla\phi$ Gibbs measures $\{ \mu_{L,\delta}\}_{L\ge 1}$ is tight, and every subsequential limit is a translation invariant and ergodic $\nabla\phi$ Gibbs measure on $\Zd$.
\end{lemma}

It follows from~\cite{She} that $\mu_\delta$ is the unique translation invariant ergodic Gibbs state on $\Zd$ with zero slope,  thus $\{ \mu_{L,\delta}\}_{L\ge 1}$ converges weakly to $\mu_\delta$.  This also induces a weak convergence of the finite volume joint measure $\tilde\mu_{L,\mathcal E}^\delta$,  stated below. 

\begin{corollary}
\label{c.weak}
The sequence of the $\nabla\phi$ Gibbs measures $\{ \mu_{L,\delta}\}_{L\ge 1}$ converges weakly to $\mu_\delta$. Moreover, fix any $K, L_1\ge 1$ and $L>K\vee L_1$, if we define the following measure on $\R^{Q_{L_1}} \times \R^{\mathcal E(Q_K)}$,
\begin{equation}
\tilde\mu_{L,\mathcal E(Q_K)}^\delta (A \times B) :=
\int_B \prod_{e\in \mathcal E(Q_K)}  e^{-\delta e^{\tau_e} - e^{- \tau_e} - \frac {\tau_e}{2}}
\biggl \langle \!\indc_A \prod_{e\in E} e^{\sqrt{\delta+ (\nabla \phi(e))^2} - (\nabla \phi(e))^2e^{\tau_e} } \biggr \rangle_{\!\!\mu_{L,\delta}} \, d\tau,
\end{equation}
for every $A\subset \R^{Q_{L_1}},  B\subset  \R^{\mathcal E(Q_K)}$. 
Then $\tilde\mu_{L,\mathcal E(Q_K)}^\delta$ converges weakly to $\tilde\mu_{\mathcal E(Q_K)}^\delta$,  defined by~\eqref{e.couple}, as $L\to \infty$.

\end{corollary}

We also notice that the ergodicity of the infinite-volume Gibbs measure $\mu_\delta$ also implies the ergodicity of $\tilde\mu_\delta$.  
\begin{lemma}
\label{l.ergodic}
$\tilde\mu_\delta$ is a translation-invariant, ergodic Gibbs measure on $\Omega \times \R^{\mathcal E(\Zd)}$. 
\end{lemma}

The proof follows from the argument in~\cite{BiSp} which uses the following explicit formula of the regular conditional distributions. 
\begin{lemma}
\label{l.cond}
Define the $\sigma$-field $\mathcal F:= \sigma\{ \nabla\phi(e): e\in \mathcal E(\Zd)\}$.  Then for $\tilde\mu_\delta$-a.e. $\phi$,   the regular conditional distribution $\tilde\mu_\delta(\cdot| \mathcal F)$,  regarded as a measure on the $\tau$ variables,  takes the product form 
\begin{equation}
\tilde\mu_\delta(d\tau| \mathcal F)
=
\prod_{e\in \mathcal E(\Zd)} e^{-\delta e^{\tau_e} - e^{- \tau_e} - \frac {\tau_e}{2}}
 e^{\sqrt{\delta+ (\nabla \phi(e))^2} - (\nabla \phi(e))^2e^{\tau_e} }\,.
\end{equation}
\end{lemma}

\begin{proof}
Recall that,  $\tilde\mu^\delta_{\mathcal E(Q_K)}$ is the restriction of $\tilde\mu_\delta$ in $Q_K$.  By definition,  conditioning on $\mathcal F_k:= \sigma\{ \nabla\phi(e): e\in \mathcal E(Q_K)\}$,  the regular conditional distribution
\begin{equation*}
\tilde\mu^\delta_{\mathcal E(Q_K)}(d\tau| \mathcal F_K)
=
\prod_{e\in \mathcal E(Q_K)} e^{-\delta e^{\tau_e} - e^{- \tau_e} - \frac {\tau_e}{2}}
 e^{\sqrt{\delta+ (\nabla \phi(e))^2} - (\nabla \phi(e))^2e^{\tau_e} }\,.
\end{equation*}
And the claim follows from standard approximation arguments.
\end{proof}

\begin{proof}[Proof of Lemma~\ref{l.ergodic}]
It is clear that $\tilde\mu_\delta$ inherits the translation invariance of $\mu_\delta$.  The ergodicity of $\tilde\mu_\delta$,  roughly speaking,  follows from the ergodicity of $\mu_\delta$ and that the regular conditional distribution given $\phi$ takes a product form.  Let $A \subset \Omega \times \R^{\mathcal E(\Zd)}$ be a translation invariant event.  We claim that $\tilde\mu_\delta(A) = 0$ or $1$ a.s. 

We first use the  ergodicity of $\mu_\delta$.  Let
\begin{equation*}
f( \nabla\phi) = \E_{\tilde\mu_\delta} [\indc_A| \mathcal F] (\nabla\phi).
\end{equation*}
Then use the translation invariance,  for every $x\in\Zd$,  let $\theta_x: \R^{\mathcal E(\Zd)} \to \R^{\mathcal E(\Zd)}$ be the translation operator defined by $\theta_x F(\phi) := F(\phi(x+\cdot))$. We then have 
\begin{equation*}
f( \theta_x \nabla\phi) = \E_{\tilde\mu_\delta} [\indc_A| \mathcal F] (\theta_x \nabla\phi)
=
\E_{\tilde\mu_\delta} [ \indc_A \cdot \theta_x^{-1}| \mathcal F] (\nabla\phi)
=
f( \nabla\phi) \,.
\end{equation*}
Since the restriction of $\tilde\mu_\delta$ to the $\phi$ variable is  $\mu_\delta$,  which is ergodic, we have $ f = c$ a.s.,  for some $c\in [0,1]$. 

We claim that $c=0$ or $1$.  By an approximation argument,  for each $N\in\N$,  there exists an event $A_N \in \sigma\{ \nabla\phi(e),  \tau(e),  e\in \mathcal E(Q_N)\}$,  such that $\tilde\mu_\delta(A_N \Delta A) \to 0$ as $N\to \infty$.  In particular,  for every $x\in\Zd$,
\begin{equation}
\label{e.approx}
\| \indc_{A_N} \indc_{\theta_x^{-1} A_N}  -\indc_{A} \indc_{\theta_x^{-1} A} \|_{L^1(\tilde\mu_\delta)}
\to 0\,.
\end{equation}
Since $A$ is translation invariant,  we have $\indc_{A} = \indc_{A} \indc_{\theta_x^{-1}\cdot A}$. Thus $\| \indc_{A_N} \indc_{\theta_x^{-1}\cdot A_N}  -
\indc_{A} \|_{L^1(\tilde\mu_\delta)}
\to 0$.
For $|x|> 2N+1$,  by Lemma~\ref{l.cond},  the regular conditional distribution is a product measure,  thus $A_N$ and $\theta_x^{-1}\cdot A_N$ are conditionally independent,  therefore
\begin{equation*}
\E_{\tilde\mu_\delta} [ \indc_{A_N} \indc_{\theta_x^{-1}\cdot A_N} | \mathcal F]
=
 \E_{\tilde\mu_\delta} [ \indc_{A_N}  | \mathcal F]
 \E_{\tilde\mu_\delta} [\indc_{\theta_x^{-1}\cdot A_N} | \mathcal F] 
\,.
 \end{equation*}
 Applying~\eqref{e.approx},  we see that by sending $N\to \infty$,  the left side above converges in $L^1(\tilde\mu_\delta)$ to $c$,  while the right side converges to $c^2$.  This implies $c=0$ or $1$,  and we conclude the ergodicity.
\end{proof}

We define the infinite-volume~$\tau$-measure as the second marginal of $\tilde\mu^\delta$: for every $B\subseteq \R^{\mathcal E(\Zd)}$,  
$$
\mu_\tau^\delta(B) := \tilde\mu^\delta (\Omega \times B)\,.
$$
Formally,  we may also integrate out an infinite-volume Gaussian measure and write the $\tau$-marginal measure by 
\begin{equation}
\label{e.mutau}
d\mu_\tau^\delta := \frac 1 {Z_{\tau,\delta}} \times \prod_{e\in \mathcal E(\Zd)}  e^{-\delta e^{\tau_e} - e^{- \tau_e} - \frac {\tau_e}{2}} 
  \det D(\tau)^{-\frac 12}  \, d\tau\,,
\end{equation}
where $D(\tau)$ is the matrix defined by  the quadratic form
\begin{equation}
\label{e.dtau}
(f, D(\tau) f):= \sum_{e \in \mathcal E(\Zd)}   (\nabla f (e))^2 e^{\tau_e}\,,
\end{equation}
and $Z_{\tau,\delta}$ is the normalizing constant so that $\mu_\tau^\delta$ is a probability measure.  It follows from Lemma~\ref{l.ergodic} that $\mu_\tau^\delta$ is a translation-invariant, ergodic Gibbs measure on $\R^{\mathcal E(\Zd)}$.  Moreover,  it follows from Corollary \ref{c.weak} and an approximation argument that the finite dimensional distribution of $\mu_{L,\tau}^\delta$ converges to $\mu_\tau^\delta$.

\smallskip

To derive an elliptic PDE for the infinite volume measure,  we recall the following lemma from~\cite{BiSp},  which implies that harmonic, translation covariant functions are uniquely determined by their mean with respect to the ergodic measures on the conductances. This is a simple consequence of the ergodic theorem, the proof of which we do not repeat here. 

\begin{lemma}[{\cite[Lemma 3.3]{BiSp}}]
\label{l.meanchar}
Let $\nu$ be a translation invariant, ergodic probability measure on the conductances $\tau = (\tau_e) \in \R^{\mathcal E(\Zd)}$.  Let $g:  \R^{\mathcal E(\Zd)} \times \Zd \to \R$ be a measurable function which satisfies the following: 
\begin{itemize}
\item $g$ is~$\tau$-harmonic $\nu$-a.s,  that is,  for every $x\in \Zd$ and almost every $\tau$,
$$
\sum_{y\sim x} \tau_{(x,y)} (g(\tau, y) -g(\tau, x) )=0.
$$
\item shift-invariance:  for every $x\in \Zd$, $i=1,\cdots,d$, we have 
$$
g(\tau, x+e_i) - g(\tau, x) = g(\theta_x \tau,  e_i)\, 
\quad \mbox{$\nu$--a.s.}
$$

\item zero mean and centering condition: $\bigl\langle g(\cdot, x)\bigr\rangle_\nu =0$ for all~$x\sim 0$ and~$g(\tau,0)=0$. 

\item square integrable in the sense that $\bigl\langle g(\cdot, x)^2\bigr\rangle_\nu<\infty$ for all $x\sim 0$.

\item square integrable gradient:  $\bigl\langle \sum_{y\sim x} \tau_{(x,y)} (g(\tau, x) - g(\tau,y))^2\bigr\rangle_\nu<\infty$.
\end{itemize}
Then we have that~$g(\cdot, x)=0$~$\nu$--a.s. for every~$x\in\Zd$. 
\end{lemma}

We remark here that we did not assume the uniform ellipticity condition as required in~ \cite[Lemma 3.3]{BiSp}. The proof of that lemma is based on an observation that the~$L^2$ space of the vector valued functions, defined via (in the notation of~\cite{BiSp}) the norm~$\|g\|_{L_{\mathrm{vec}}^2(\nu)}:=\langle \,\sum_{x\sim 0} \tau_{(0,x)} g(\tau, x)^2\rangle_\nu$, can be decomposed as a direct sum of the~$L^2$ space of the gradient functional, and the space of harmonic functions with respect to the~$\tau$-conductance. The uniform ellipticity condition was used only to ensure that the gradient of~$g$ belongs to the space $L_{\mathrm{vec}}^2(\nu)$. For a general probability measure~$\nu$ on the conductance (which is not necessary uniformly elliptic), the finiteness of the $L_{\mathrm{vec}}^2(\nu)$ norm can be checked directly from the $\tau$-harmonicity of $g$ and the fact that the conductances are nonnegative.

\smallskip

We are now able to characterize the conditional law of the $\phi$-field given the $\tau$'s.

\begin{lemma}
\label{l.conditional}
Consider  the $\sigma$-field $\mathcal F:= \sigma\{ \tau(e): e\in \mathcal E(\Zd)\}$.  For $\tilde \mu^\delta$-a.e. $\tau$,  the conditional law $\tilde \mu^\delta(\cdot|\mathcal F)(\tau)$,  regarded as a measure on the set of configurations $\{(\phi(x)) \in \R^{\Zd} : \phi(0) =0\}$,  is a multivariate Gaussian with mean zero,
$$
\E_{\tilde \mu^\delta} [\phi(x) |\mathcal F] (\tau)=0,  \quad  x\in \Zd,
$$
and covariance matrix given by $D(\tau)^{-1}$,  where $D(\tau)$ is defined in \eqref{e.dtau}.
\end{lemma}

We will present a proof of the lemma at the end of Section \ref{s.perco},  using Lemma \ref{l.meanchar} and the Brascamp-Lieb inequality presented in that section.  

Recall that $\f_R := \f(\frac{\cdot}{R})$,  applying the lemma above, we may now rewrite 
\begin{align}
\var_{\mu^\delta}[F_R] 
&
= \var_{\mu^\delta} \Bigl[R^{-\frac d2}\sum_{x\in\Zd}\sum_{i=1}^d f_i\left(\frac xR\right) \left( \phi(x+e_i)-\phi(x) \right) \Bigr]
\notag \\ & 
= R^{-d} \sum_{x,y\in \Zd}\left\langle \nabla \cdot \f_R (x) G_\a(x,y)  \nabla \cdot \f_R (y) \right\rangle_{\mu_\tau^\delta}\,.
\end{align}
Here $\a$ is a diagonal matrix with $\a(e,e) = e^{\tau_e}$, and the only randomness inside the expectation is $G_\a(x,y) $, defined to be a solution of 
\begin{equation*}
\nabla \cdot \a \nabla G_\a(\cdot, y) = \delta_y \,.
\end{equation*}
If we define $u$ to be the solution of the PDE
\begin{equation}
\label{e.rwrcf}
\nabla \cdot \a \nabla u = \nabla\cdot \f_R\,,
\end{equation}
then we may rewrite the identity as 
\begin{equation}
\var_{\mu^\delta}[F_R]  =\left\langle (\f_R, \nabla u)_R \right\rangle_{\mu_\tau^\delta}\,,
\end{equation}
where the inner product $(\cdot, \cdot)_R$ is defined in~\eqref{e.inner}.
This justifies the elliptic PDE connection in infinite volume.

The convergence of the variance of $F_R$ would then follow from the homogenization of~\eqref{e.rwrcf} to a deterministic PDE
\begin{equation}
\nabla \cdot \ahom \nabla \bar u = \nabla\cdot \f
\quad \mbox{in} \ \Rd\,,
\end{equation}
where $\ahom$ is a deterministic diagonal matrix with positive entries.  This is established in Theorem~\ref{t.homogenize},  from which we conclude Theorem~\ref{t.CLT}.

\section{Percolation estimates}
\label{s.perco}

In this section we obtain percolation estimates for the conductances $\{\a(e)\}$, which implies the probability that the origin belongs to a large cluster of bad conductances (very large or very close to zero) is very small.  The main estimates proved in this section are Corollary~\ref{c.mgf},  which shows all negative moments of the conductance $\a(e)$ are finite,  and Proposition~\ref{p. percolation},  which implies the set of edges where $\a(e)$ is large is stochastically dominated by a subcritical percolation cluster (which can be made very subcritical if we allow a higher ellipticity contrast). 

\subsection{Percolation bounds in finite volume}

We first recall some classical results about the determinant of Laplacian.  Recall that a spanning tree of $G$ is a connected graph using all vertices of $G$ in which there are no circuits.  A spanning tree of $G$ with wired boundary condition is a spanning tree of $G$ where the exterior of $G$ is identified to a single vertex.  

\begin{lemma}[Matrix-Tree Theorem~\cite{Kir}]
\label{l.tree}
Let $\mathcal S$ be the set of spanning trees in $Q_L$ with wired boundary condition,  and denote by $D_L$ the matrix defined in~\eqref{e.dl}.  Then 
\begin{equation}
\det D_L(\tau) = \sum_{\mathcal T \in \mathcal S} \prod_{e\in \mathcal T } e^{\tau_e}.
\end{equation}
\end{lemma}

The following immediate corollary was obtained in~\cite{BS}.

\begin{lemma}
\label{l.logdet}
As a function in $\tau$, $\frac{1}{\sqrt{\det D_L(\tau) }} $ is log-concave. 
\end{lemma}
\begin{proof}
Since $\det D_L(\tau) $ is a sum of exponential functions in $\tau$, it is log-convex. 
\end{proof}

We also recall that the marginal of a log-concave function is log-concave (see,  e.g.,  Corollary 3.5 of~\cite{BL}).
\begin{lemma}
\label{l.convexint}
Let $F: \R^{m+n} \to \R$ be log-concave, $x\in \R^m$ and $y\in\R^n$. Define $G: \R^m \to \R$ by 
\begin{equation}
G(x) := \int_{\R^n} F(x,y) \, dy \,.
\end{equation}
Then $G$ is log-concave. 
\end{lemma}

Given a function $\mathsf s: \mathcal E(Q_L) \to [0,\frac 12)$, we define the tilted measure 

\begin{equation}
\label{e.mutautilt}
d\mu_{L,\tau}^\delta(\mathsf s)  := \frac{1}{Z_{L,\delta, \tau} (\mathsf s)} \prod_{e\in \mathcal E(Q_L)}  e^{-\delta e^{\tau_e} - e^{- \tau_e} - \frac {\tau_e}{2} + s_e \tau_e} 
  \det D_L(\tau)^{-\frac 12}  \, d\tau\,.
\end{equation}
where 
\begin{equation}
\label{e.Ztilt}
Z_{L,\delta, \tau} (\mathsf s) := \int \prod_{e\in \mathcal E(Q_L)}  e^{-\delta e^{\tau_e} - e^{- \tau_e} - \frac {\tau_e}{2} + s_e \tau_e} 
  \det D_L(\tau)^{-\frac 12}  \, d\tau\,.
\end{equation}

The next lemma gives bounds for the expectation under the tilted measure. 

\begin{lemma}
\label{l.mean}
Fix $k\in\N$. There exists a constant $C_k<\infty$, such that for any function $\mathsf s: \mathcal E(Q_L) \to [-k,\frac 14)$, and any $e_0\in \mathcal E(Q_L)$, $-C_k \leq \langle \tau_{e_0} \rangle_{\mu_{L,\tau}^\delta(\mathsf s)} \leq C_k$.
\end{lemma}

\begin{proof}
We omit the dependence of $Z$ on $\delta, \tau, \mathsf s$ in the proof.  Fix an $e_0\in \mathcal E$ and we denote by $g(\tau_{e_0})$ the marginal density of~\eqref{e.mutautilt} at $ \tau_{e_0}$, namely 
\begin{equation}
g(\tau_{e_0}) := \frac 1{Z_{L}}  e^{-\delta e^{\tau_{e_0}} - e^{- \tau_{e_0}} - \frac {\tau_{e_0}}{2} + s_{e_0} \tau_{e_0}} Z_L(\tau_{e_0})\,,
\end{equation}
where we denote  
\begin{equation}
Z_L(\tau_{e_0}):=  \int \prod_{e\neq e_0}  e^{-\delta e^{\tau_e} - e^{- \tau_e} - \frac {\tau_e}{2} + s_e \tau_e} 
  \det D_L(\tau)^{-\frac 12}  \prod_{e\neq e_0} d\tau_e\,.
\end{equation}
Set $f(\tau_{e_0}): = - \log \frac{Z_L(\tau_{e_0})}{Z_{L}}$, so that 
$$
g(\tau_{e_0}) =  e^{-\delta e^{\tau_{e_0}} - e^{- \tau_{e_0}} - \frac {\tau_{e_0}}{2} + s_{e_0} \tau_{e_0} -f(\tau_{e_0})}\,.
$$
 We observe that
\begin{equation}
Z_L(\tau_{e_0})=  \int 
\exp\left(-\sum_{e\neq e_0 } \Bigl(  \delta  e^{\tau_e}+ e^{-\tau_e}+\Bigl(\frac 12-s_e\Bigr) \tau_e\Bigr) -\frac 12 \log  \det D_L(\tau)
  \right) \prod_{e\neq e_0} d\tau_e\,.
\end{equation}
By Lemma~\ref{l.logdet},  the integrand is log concave, and thus by Lemma~\ref{l.convexint}, $f$ is a convex function of $\tau_{e_0}$. We may compute the first derivative of $f$ explicitly as 
\begin{align}
\label{e.f'}
f'(\tau_{e_0}) & 
= -\frac{\partial_{\tau_{e_0}}Z_L(\tau_{e_0})}{Z_L(\tau_{e_0})} 
\notag \\ &
= \frac{1}{Z_L(\tau_{e_0})}
\int  \frac 12  \partial_{\tau_{e_0}}  \log \det D_L(\tau) \prod_{e\neq e_0} 
e^{-\delta e^{\tau_e} - e^{- \tau_e} - \frac {\tau_e}{2} + s_e \tau_e} 
\det D_L(\tau)^{-\frac 12}  \prod_{e\neq e_0} d\tau_e
\,.
\end{align}
We claim that $f'$ is nonnegative and bounded above by constants independent of~$L$.  Applying Lemma~\ref{l.tree}, we see that
  \begin{equation}
 1 \ge \partial_{\tau_{e_0}}  \log \det D_L(\tau) = 
  \frac{\sum_{\mathcal T: e_0 \in \mathcal T}\prod_{e\in \mathcal T } e^{\tau_e} }{\sum_{\mathcal T \in \mathcal S}\prod_{e\in \mathcal T } e^{\tau_e}}
  = \P_\mathcal T (e_0 \in \mathcal T) \ge 0\,,
  \end{equation}
  where for given $\{e^{\tau_e}\}_{e\in \mathcal E(Q_L)}$,  $\mathcal T$ has the law of a random spanning tree with wired boundary condition in $Q_L$,  weighted by $\prod_{e\in \mathcal T } e^{\tau_e}$. 
  Thus by \eqref{e.f'},  $1/2 \ge f'(\tau_{e_0})  \ge 0$.  
  We may Taylor expand $f$ which yields
  \begin{equation}
  f(\tau_{e_0}) = f(0) + f'(\xi)\tau_{e_0},
  \end{equation}
  for some $\xi \in [0, \tau_{e_0}]$. Therefore 
  \begin{equation}
  \label{e.g}
  g(\tau_{e_0}) =  e^{-\delta e^{\tau_{e_0}} - e^{- \tau_{e_0}} - \frac {\tau_{e_0}}{2} + s_{e_0} \tau_{e_0}}  
  e^{-f(0) - f'(\xi)\tau_{e_0}}.
  \end{equation}
 We claim that $f(0)$ is uniformly bounded (independent of $L$).  Notice that, for every $u,v\in[-1,1]$ and $s_{e_0}\in [-k, \frac 14]$, there exists $C_k<\infty$ such that 
  $$
  \frac {g(u)}{g(v)} \leq C_k e^{-f(u) +f(v)}
  \leq
  C_k e^{\max_{x\in [-1,1]}|f'(x)| |u-v|}
  \leq 
  C_ke^{\frac 12 |u-v|}\,,
  $$
  which implies $g$ has bounded ratio in $[-1,1]$.  Since $\int_{-1}^1 g(\tau) \, d\tau \leq \int_{-\infty}^\infty g(\tau) \, d\tau =1$,  we conclude that there exists $C'_k$ such that  $g(0)\leq C'_k$ and $e^{-f(0) } \leq C'_k$.  
Fix $s_{e_0} \in [-k,\frac14]$,  combine ~\eqref{e.g} with $e^{-f(0) } \leq C'_k$ and $1/2 \ge f'(\tau_{e_0})  \ge 0$,  we see that the marginal density of $\tau_{e_0}$ decays exponentially as $\tau_{e_0}\to \infty$ or $\tau_{e_0}\to -\infty$ (with the exponential rate depends continuously on $s_{e_0}$). Therefore there exists $C''_k<\infty$ such that  $-C''_k \leq \langle \tau_e \rangle_{\mu_{L,\tau}^\delta(\mathsf s)} \leq C''_k$.
\end{proof}

\begin{corollary}
\label{c.mgffinite}
Fix $k,L\in\N$, there exists $C_k<\infty$, such that for every $e\in \mathcal E(Q_L)$,
\begin{equation}
\langle  \a(e)^{-k} \rangle_{\mu_{L,\tau}^\delta}
=
\langle  e^{-k\tau_e} \rangle_{\mu_{L,\tau}^\delta} \leq C_k.
\end{equation}
\end{corollary}
\begin{proof}
For every $s\in [-k,0]$,  we apply Lemma~\ref{l.mean} with $\mathsf s= \indc_{e} (\cdot) s_e$  to obtain
\begin{equation*} 
\log \langle  e^{-k\tau_e} \rangle_{\mu_{L,\tau}^\delta} 
 = \int_0^{-k} \frac{d}{ds}  \log \langle  e^{s\tau_e} \rangle_{\mu_{L,\tau}^\delta}  \, ds
 \leq \left| \int_0^{-k} \langle \tau_e \rangle_{\mu_{L,\tau}^\delta(\mathsf s)} \, ds\right|
 \leq C_k\,. \qedhere 
 \end{equation*}
 \end{proof}
 
We are now able to obtain the percolation estimate. 

\begin{proposition}
\label{p.percolationfinite}
There exists $M_0<\infty$, such that for any $M>M_0$, $1<K<\infty$ and any finite collection of edges $\mathcal C \subset  \mathcal E(Q_{L})$, we have 
\begin{equation}
\label{e.bigcondL}
\biggl \langle \prod_{e \in\mathcal{C} }  \indc_{\{\tau_e \ge KM\}} \biggr \rangle_{\!\!\mu_{L,\tau}^\delta}
\leq e^{-\alpha(K,M) |\mathcal{C} |} 
\end{equation}
and
\begin{equation}
\label{e.smallcondL}
\biggl \langle \prod_{e \in\mathcal{C} }  \indc_{\{\tau_e \le -KM\}} \biggr \rangle_{\!\!\mu_{L,\tau}^\delta}
\leq e^{-\alpha(K,M) |\mathcal{C} |} \,.
\end{equation}
Moreover, $\alpha(K,M) \to K$ as $M\to \infty$. 
\end{proposition}
In particular, if we choose $K,  M$ sufficiently large then the  edges in $Q_{L}$ with conductances $\a(e) = e^{\tau_e}$ larger than $e^{KM}$ or smaller than $e^{-KM}$  are dominated by a very subcritical percolation cluster. 

\begin{proof}
Notice that
\begin{equation}
\biggl \langle \prod_{e \in\mathcal{C}}  \indc_{\{\tau_e \ge KM\}} \biggr \rangle_{\!\!\mu_{L,\tau}^\delta}
= \biggl \langle \prod_{e \in\mathcal{C}} \indc_{\{\tau_e \ge KM\}} e^{-\frac{\tau_e}{M}}e^{\frac{\tau_e}{M}}\biggr \rangle_{\!\!\mu_{L,\tau}^\delta}
\leq e^{-K|\mathcal{C}|} \biggl  \langle \prod_{e \in\mathcal{C}} e^{\frac{\tau_e}{M}}\biggr \rangle_{\!\!\mu_{L,\tau}^\delta}\,.
\end{equation}
The last quantity can be written in terms of the tilted partition function ~\eqref{e.Ztilt} by
\begin{equation}
\left\langle \prod_{e \in\mathcal{C}} e^{\frac{\tau_e}{M}}\right\rangle_{\mu_{L,\tau}^\delta} = 
\frac{Z(\frac 1M \indc_{\mathcal{C}}(\cdot)) }{Z} \,.
\end{equation}
Let $\mathsf s_e := \frac 1M \indc_{\{e \in\mathcal{C}\}}$.  By differentiation, 
$$
\frac d{ds} \log \left\langle \prod_{e \in\mathcal{C}} e^{s\tau_e}\right\rangle_{\mu_{L,\tau}^\delta}
= 
\sum_{e \in\mathcal{C}} \langle \tau_e \rangle_{\mu_{L,\tau}^\delta(\mathsf s)} \,,
$$
thus
\begin{equation}
\log \frac{Z(\frac 1M \indc_{\mathcal{C}}(\cdot) )}{Z} = \sum_{e \in\mathcal{C}}  \int_0^{\frac{1}{M}}
\langle \tau_e \rangle_{\mu_{L,\tau}^\delta(\mathsf s)} \,ds \,.
\end{equation}
Lemma~\ref{l.mean} implies there is an absolute constant $C$,
\begin{equation}
\sup_{s\in [-1, \frac14]} \langle \tau_e \rangle_{\mu_{L,\tau}^\delta(\mathsf s)} \leq C\,.
\end{equation}
Therefore
\begin{equation}
\frac{Z(\frac 1M \indc_{\mathcal{C}}(\cdot) )}{Z} \leq \exp\left( \frac{C}{M} |\mathcal{C}|\right)\,.
\end{equation}
Take $M_0=C$ and  $M>M_0$ this implies bound~\eqref{e.bigcondL}. The proof of~\eqref{e.smallcondL} is very similar, by using 
\begin{equation}
\left\langle \prod_{e \in\mathcal{C}}  \indc_{\{ \tau_e \le -KM \}} \right\rangle_{\mu_{L,\tau}^\delta}
\leq e^{-K|\mathcal{C}|} \left\langle \prod_{e \in\mathcal{C}} e^{-\frac{\tau_e}{M}}\right\rangle_{\mu_{L,\tau}^\delta}
\end{equation}
and 
\begin{equation}
\frac{Z(-\frac 1M\indc_{\mathcal{C}}(\cdot) ) }{Z} \leq \exp\left( \frac{C}{M} |\mathcal{C}|\right)\,.
\end{equation}
This completes the proof.
\end{proof}

\subsection{Tightness}
\label{s.tight}
In this section we prove Lemma~\ref{l.tight} that establishes the tightness of the sequence of the Gibbs measures $\{\mu_{L,\delta}\}$.

We start by  stating a Brascamp-Lieb inequality, which gives a Gaussian domination for moments of linear functionals of $\phi$.  This was proved in~\cite{BS} for the case $V_\delta(x) = \sqrt{x^2+\delta}, \delta>0$ and we follow the same line of argument here. 

\begin{proposition}
\label{p.BLfinite}
Let $\phi$ be sampled from the continuous solid-on-solid Gibbs measure in $Q_L$ with zero boundary condition,  and let $v: Q_L \to \R$ be a mean zero test function: $\sum_{x\in Q_L}v(x) =0$.  Then for every $k\in\N$,  there exists $C_k <\infty$,  such that 
\begin{equation}
\Biggl\langle \biggl(\sum_{x\in Q_L} \phi(x) v(x) \biggr)^{2k} \Biggr\rangle_{\mu_{L,\delta}} \leq 
C_k \Biggl( \sum_{(i,j)\in \mathcal E(Q_L)} ((\Delta^{-1} v)(i)-  (\Delta^{-1} v)(j))^2\Biggr)^{\!k}.
\end{equation}
\end{proposition}

\begin{proof}
By using the joint measure for $\phi$ and $\tau$~\eqref{e.jointZ},  integrate out the Gaussian marginal, and apply the Wick's theorem to compute the Gaussian moments, we have

\begin{equation*}
\left\langle \biggl(\sum_{x\in Q_L} \phi(x) v(x) \biggr)^{2k} \right\rangle_{\!\!\mu_{L,\delta}}
= (2k-1)!! \left\langle \biggl(\sum_{x,y\in Q_L}  v(x)G_{\a, L}(\tau)(x,y)v(y)\biggr)^k \right\rangle_{\!\!\mu_{L,\tau}^\delta},
\end {equation*}
where $G_{\a, L}(\tau) = D_L(\tau)^{-1}$,  and $D_L(\tau)$ is defined in~\eqref{e.dl}.  We may write $D_L (\tau)=\nabla\cdot A^2\nabla$,  were $A$ is a diagonal matrix with entries $A(e,e) = e^{\tau_e/2}$.  Using integration by parts and denoting~$G_L:= \Delta^{-1}$, we find that 
\begin{align*}
\sum_{x,y\in Q_L}  v(x)G_{\a, L}(\tau)(x,y)v(y)
= 
\left(  v, \Delta G_L \cdot G_{\a, L}v\right)_{Q_L}
=
\left(  \nabla G_L v, \nabla G_{\a, L}v\right)_{Q_L} \\
= 
\left( A^{-1} \nabla G_L v, A\nabla G_{\a, L}v\right)_{Q_L}
\leq
\left( A^{-1} \nabla G_L v, A^{-1} \nabla G_L v\right)_{Q_L}^\frac 12
\left( A\nabla G_{\a, L}v,  A\nabla G_{\a, L}v\right)_{Q_L}^\frac 12\\
=
\biggl( \sum_{e=(i,j)\in \mathcal E(Q_L)} ((\Delta^{-1} v)(i)-  (\Delta^{-1} v)(j))^2 e^{-\tau_e}
\biggr)^\frac 12
\left(  v, G_{\a, L}v\right)_{Q_L}^\frac 12,
\end {align*}
where we used $\nabla\cdot A^2\nabla G_{\a, L} = \delta$ in the last line.  
Therefore by absorbing the last term to the left,
\begin{equation*}
 \left(\sum_{x,y\in Q_L}  v(x)G_{\a, L}(\tau)(x,y)v(y)\right)^k
 \leq
 \Biggl( \sum_{e=(i,j)\in \mathcal E(Q_L)} ((\Delta^{-1} v)(i)-  (\Delta^{-1} v)(j))^2 e^{-\tau_e}
\Biggr)^k\,.
\end {equation*}
To conclude the proposition,  we expand the $k^{th}$ power in the right side,  use the H\"older inequality to bound
$$
\langle e^{-\tau_{e_1}} \cdots e^{-\tau_{e_k}} \rangle_{\mu_{L,\tau}^\delta}
\leq
\left( \langle e^{-k\tau_{e_1}} \rangle_{\mu_{L,\tau}^\delta}  \cdots \langle e^{-k\tau_{e_k}} \rangle_{\mu_{L,\tau}^\delta}  \right)^\frac 1k.
$$
and apply the exponential moment bound Corollary~\ref{c.mgffinite}. 
\end{proof}

\begin{proof}[Proof of Lemma~\ref{l.tight}]
Lemma~\ref{l.tight} now follows form standard arguments.  Define, for any $r>0$, the space
\begin{equation}
\Omega_{r,M} := \biggl\{ \phi: \Zd \to \R: \| \nabla \phi \|_r^2 
:= \sum_{x\in\Zd}\sum_{i=1}^d (\nabla_i\phi(x))^2  e^{-2rx} \leq M \biggr\}\,,
\end{equation}
which forms compact subsets of $\Omega$, the space of gradient functions. We may compute 

$$
\mu_{L,\delta} (\Omega_{r,M}^c) = \langle \indc_{\{ \| \nabla \phi \|_r > M \} } \rangle_{\mu_{L,\delta}}
\leq \frac{1}{M^2}  \langle \| \nabla \phi \|_r^2 \rangle_{\mu_{L,\delta}} 
=\frac{1}{M^2}  \sum_{x\in Q_L}\sum_{i=1}^d \langle (\nabla_i \phi(x))^2 \rangle_{\mu_{L,\delta}}  e^{-2r|x|} \,.
$$
We apply Proposition~\ref{p.BLfinite} with $v = \delta_x - \delta_{x+e_i}$ to obtain
$$
\langle (\nabla_i  \phi(x))^2 \rangle_{\mu_{L,\delta}} 
\leq
C \sum_{(j,k)\in \mathcal E(Q_L)} ((\Delta^{-1} v)(j) -  (\Delta^{-1} v)(k))^2\,.
$$
 Therefore using integration by parts,
 \begin{equation}
 \langle (\nabla_i  \phi(x))^2 \rangle_{\mu_{L,\delta}} 
 \leq C   (\delta_x - \delta_{x+e_i}, \Delta^{-1}(\delta_x - \delta_{x+e_i}))
 < \infty
\,.
 \end{equation}
This implies
\begin{equation} 
\mu_{L,\delta} (\Omega_{r,M}^c) \leq
\sup_{x\in Q_L, i=1.\cdots, d} \langle (\nabla_i  \phi(x))^2 \rangle_{\mu_{L,\delta}} 
\frac{d}{M^2}  \sum_{x\in Q_L}   e^{-2r|x|} \,,
\end{equation}
so that for $M$ suffciently large, $\mu_{L,\delta} (\Omega_{r,M}^c) \leq \eps$ for all $L\ge 1$. This gives the tightness of $\mu_{L,\delta}$.  
\end{proof}

\subsection{Percolation bounds in infinite volume}

Combining Corollary~\ref{c.weak} and Proposition~\ref{p.percolationfinite}, we conclude that for any fixed $L\ge L_1\ge 1$, the law of $\{\tau_e\}_{e\in \mathcal E(Q_{L_1})}$ under $\mu_{L,\tau}^\delta$ converges in distribution as $L\to\infty$. In particular, this implies the following statement. 

\begin{proposition}
\label{p. percolation}

There exists $M_0<\infty$, such that for any fixed  $M>M_0$, $K\ge 1$, and any finite collection of edges $\mathcal C $, we have 
\begin{equation}
\label{e.bigcond}
\left\langle \prod_{e \in\mathcal{C}}  \indc_{\{\tau_e \ge KM\}} \right\rangle_{\!\!\mu_{\tau}^\delta}
\leq e^{-\alpha(K,M) |\mathcal{C} |} 
\end{equation}
and, likewise,
\begin{equation}
\label{e.smallcond}
\left\langle \prod_{e \in\mathcal{C} }  \indc_{\{\tau_e \le -KM\}} \right\rangle_{\!\!\mu_{\tau}^\delta}
\leq e^{-\alpha(K,M) |\mathcal{C} |} \,.
\end{equation}
Moreover, $\alpha(K,M) \to K$ as $M\to \infty$. 
\end{proposition}

The inverse moments bounds and Brascamp-Lieb inequality also extends to infinite volume by passing to the limit~$L\to\infty$.  

\begin{corollary}
\label{c.mgf}
For every~$k\in\N$, there exists $C_k<\infty$ such that, for every $e\in \mathcal E(\Zd)$,
\begin{equation}
\bigl \langle  \a(e)^{-k} \bigr \rangle_{\mu_{\tau}^\delta}
=
\bigl \langle  e^{-k\tau_e} \bigr \rangle_{\mu_{\tau}^\delta} \leq C_k\,.
\end{equation}
\end{corollary}

\begin{proposition}
\label{p.BL}
Let $\phi$ be sampled from the infinite-volume continuous solid-on-solid measure,  and let $v: \Zd \to \R$ be a mean zero test function with finite support: $\sum_{x\in \Zd}v(x) =0$.  Then for every $k\in\N$,  there exists $C_k <\infty$,  such that 
\begin{equation}
\left\langle \biggl (\sum_{x\in \Zd} \phi(x) v(x) \biggr )^{\!\!2k} \right\rangle_{\!\!\mu_{\tau}^\delta} \leq 
C_k \Biggl ( \sum_{(i,j)\in \mathcal E(\Zd)} ((\Delta^{-1} v)(i) -  (\Delta^{-1} v)(j))^2\Biggr )^{\!k}.
\end{equation}
\end{proposition}

We now apply the Brascamp-Lieb inequality to prove Lemma \ref{l.conditional}.

\begin{proof}[Proof of Lemma \ref{l.conditional}]
By directly checking \eqref{e.couple},  we see that the conditional measure is a multivariate Gaussian with conductance $(\tau_e)$,  this leads to the covariance $D(\tau)^{-1}$.  To identify the mean,   we define a function $g:  \R^{\mathcal E(\Zd)} \times \Zd \to \R$ by 
$g(\tau,  0)=0$,  and for every $x\in\Zd$ and $i=1,\cdots,  d$,
$$
g(\tau,  x+e_i) - g(\tau,  x) = \E_{\tilde \mu^\delta}[\nabla \phi(x+e_i)| \mathcal F] (\tau).
$$
Since the right side is a gradient,  $g$ is a well-defined function.  A direct computation of the conditional mean ensures that $g$ is harmonic with respect to the conductance $\tau$:
$$
\sum_{y\sim x} \tau_{(x,y)} (g(\tau, y) -g(\tau, x) )
= 
\E_{\tilde \mu^\delta} \biggl( \sum_{y\sim x} \tau_{(x,y)}  (\phi(y)-\phi(x)) | \mathcal F\biggr) (\tau)
= 0,
$$
for a.s-$\tau$.  The shift-invariance of $g$ follows from the translation invariance of $\mu^\delta$.  The fact that the $(\tau_e)$ is nonnegative and $g$ is $\tau$-harmonic implies for every $x\in\Zd$,
$$
\biggl\langle \sum_{y\sim x} \tau_{(x,y)} (g(\tau, x) - g(\tau,y))^2\biggr\rangle_{\!\tilde \mu^\delta}
<\infty.
$$
To check the square integrability,  we apply the Jensen's and the Brascamp-Lieb inequalities which yields for every $x\sim 0$,
$$
\langle g(\tau, x)^2\rangle_{\tilde \mu^\delta}
=
\E_{\tilde \mu^\delta} \bigl[ \E_{\tilde \mu^\delta} [\phi(x) -\phi(0)| \mathcal F] \bigr]^2
\le
\E_{\tilde \mu^\delta} \bigl[ \E_{\tilde \mu^\delta} [(\phi(x) -\phi(0))^2| \mathcal F] \bigr]
= 
\langle (\phi(x) -\phi(0))^2\rangle_{\tilde \mu^\delta}
<\infty.
$$
Finally,  since $\tilde \mu^\delta$ has zero slope,  
$$
\langle g(\cdot, x)\rangle_{\tilde \mu^\delta} = \langle \phi(x) - \phi(0)\rangle_{\tilde \mu^\delta}
=0, \quad
\forall x\sim 0.
$$
Since $g$ satisfies all conditions of Lemma \ref{l.meanchar},  we have $g(\cdot, x)=0$ $\tau$-a.s.
\end{proof}

\section{Qualitative homogenization}
\label{s.homog}
\subsection{Coarse-grained coefficients}

In this section, we introduce two subadditive energy quantities related to the
variational formulation of the PDE~\eqref{e.rwrcf}.  These quantities are analogous to the ones found in~\cite{AKMbook,AKbook}.  They
represent, respectively, the energy of the Dirichlet problem with an approximately affine boundary data and that of the Neumann problem with an approximately constant boundary 
flux (some modification will be made on the edges with very large or very small conductances).  These quantities are approximately \emph{subadditive} and therefore converge as the side length of the cube becomes large. The main focus of this section is to show that these quantites converge to quadratic functions with positive definite coefficients by applying the percolation estimates above.  

\subsubsection{Good cubes}
The diffusion matrix associated with the PDE~\eqref{e.rwrcf} is degenerate elliptic: in principle $\a(e)$ can be very close to $0$ or $\infty$.  In this section we construct the good cubes based on the percolation bounds in Section~\ref{s.perco},  which implies that beyond some (random) scale, the effective conductance have uniform upper and lower bounds. This allows us to prove homogenize beyond this random scale.

\smallskip

We first introduce the notion of a ``good'' cube.  A lattice animal is a connected subset of $\Z^d$ that contains the origin.  The set of lattice animals of size $n$ is denoted by $A(n)$.  The connectivity constant of the lattice animals in $\Z^d$ is given by the constant $\alpha$ such that 

\begin{equation*}
\alpha = \lim _{n\to \infty}\frac 1n \log |A(n)|,
\end{equation*}
the limit exists by subadditivity.  Let $p_0\in (0,\frac 1{2\alpha}) $ be a sufficiently small constant, to be determined later in the proof of Lemma \ref{l.goodS}.  By Proposition~\ref{p.  percolation},  there exists $M_0 < \infty$ and $K \in \N$,  such that for any finite edge set $\mathcal{C}\subset \mathcal E(\Zd)$,  
\begin{equation}
\label{e.verysub}
\biggl \langle \prod_{e \in\mathcal{C} }  \indc_{\{\tau_e \ge KM_0\}} \!\biggr \rangle_{\!\!\mu_{\tau}^\delta}
\leq p_0^{|\mathcal{C} |} .
\end{equation}
From now on we assume $\delta=0$ (and omit the dependence in $\delta$).  

Define $\mathcal C_{>}:= \{(x,y) \in \Z^d \times \Z^d: x \sim y,  \a((x,y)) > KM_0 \text{ or } \a((x,y)) < (KM_0)^{-1}\}$. For $x\in \Z^d$,  denote by $\mathcal C_{>}(x)$ the connected component of $\mathcal C_{>}$ that contains $x$.  Denote by $x\sim y$ if $y\in\mathcal C_{>}(x)$, then $\sim$ forms an equivalence relation. We write $\mathcal C_{>} = \bigcup_{x\in \Zd} \mathcal C_{>}(x)$ and $[\mathcal C_{>}] = \mathcal C_{>}\setminus \sim$.   In other words,  $[\mathcal C_{>}]$ contains a representative point in each connected component.  We also define the boundary set $\partial \mathcal C_{>}(x_0):= \{x \in \Z^d \setminus  \mathcal C_{>}: x \sim y \text{ for some } y\in  \mathcal C_{>}(x_0)\}$.  
For $n\in\N$ denote by $\cu_n := [-3^n, 3^n]\cap \Zd$.  For every $U\subset \Zd$,  define $\diam(U)$ the graph diameter of $U$.

\begin{definition}[Good cube] \label{d.good}

For $z\in\Zd$ and $n\in\N$, we say that the cube $z + \cu_n$ is a \emph{good cube} provided that the following  conditions are satisfied:
\begin{itemize}

\item Rarity of high conductances: 
\begin{align}
\label{e.d+1moment}
\diam ( \mathcal C_{>} \cap  (z+ \cu_{n+1}) )^{d+2} \leq \frac 1 {100} |\cu_n | \,,
\end{align}
where 
\begin{align*}
\diam ( \mathcal C_{>} \cap ( z+ \cu_{n+1}) )^{d+2} :=
\sum_{x\in [\mathcal C_{>} \cap ( z+ \cu_{n+1})]} (\diam ( \mathcal C_{>}(x))^{d+2}.
\end{align*}
Moreover,  for any $p>0$,  there exists $C_p <\infty$, such that 
\begin{align}
\label{e.pmoment}
\diam ( \mathcal C_{>} \cap  (z+ \cu_{n+1}) )^{p} \leq C_p |\cu_n | 
\end{align}

\item  Upper bound for the inverse moments:
for every $p\ge 1$, 
\begin{align}
\bigg(
\frac{1}{|\cu_n|}\sum_{e\in \mathcal E( \cu_n)} 
\a^{-p}(e)
\bigg)^{\!\!-\frac 1p}
\geq \frac12\left\langle 
\a^{-p}(e)
\right\rangle^{\!-\frac 1p}_{\mu_\tau},
\label{e.invmom}
\end{align}

\end{itemize}
\end{definition}

The rarity of high conductances assumption follows from the percolation bounds Proposition~\ref{p.  percolation} and~\eqref{e.verysub},  and the inverse moments assumption follows from Corollary~\ref{c.mgf} and the ergodic theorem.  Moreover,  we conclude that,  
\begin{itemize}

\item There exists a minimal scale~$\S$ which is finite~$\P$--almost surely:
\begin{align*}
\P \bigl[ \S < \infty \bigr] = 1\,.
\end{align*}

\item Cubes are good on scales larger than~$\S$:
\begin{align} \label{e.ass.minimal}
n\in \N, \ 3^n\geq \S 
\quad \implies \quad  
\cu_n \mbox{ is a good cube.}
\end{align}

\end{itemize}
Regarding the edges with high conductance,  we may use the percolation estimate~\eqref{e.verysub} to obtain the stochastic  integrability of the minimal scale at which \eqref{e.d+1moment} and \eqref{e.pmoment} hold.  Define $\S_1$ to be the minimal scale such that for all $n \in\N, \ 3^n\geq \S_1 $ ,  $\cu_n$ satisfy \eqref{e.d+1moment} and \eqref{e.pmoment}.

\begin{lemma}
\label{l.goodS}
There exists $n_0\in \N$ and $\kappa>0$, such that for every $m\in\N$,  we have
\begin{align} 
\label{e.si}
\P[\S_1 > 3^{n_0+m}] \leq e^{-\kappa 3^{m/2}}.
\end{align}

\end{lemma}

To prove the lemma, we first introduce some notations.  For a random variable $X$ and $s, \theta \in (0,\infty)$,  we write $X= \mathcal O_s(\theta)$ to mean that 
\begin{align*} 
\E\bigl[e^{(\theta^{-1}X)^s} \bigr] \le 2.
\end{align*} 
 For a random variable $X$ and $s\in (1,2], \theta \in (0,\infty)$,  we write $X= \bar{ \mathcal O}_s(\theta)$ to mean that 
 \begin{align*} 
\log \E\bigl [e^{\lambda \theta^{-1}X} \bigr] \le \lambda^2 \vee |\lambda|^{\frac s{s-1}}.
\end{align*} 

\begin{proof}
Using the percolation estimate~\eqref{e.verysub},  $ \mathcal C_> $ is stochastically dominated by a $p_0$-Bernoulli percolation cluster. It suffices to prove the stochastic integrability for $p_0$-percolation clusters.  Given $x\in\Zd$,  denote by $\mathcal C_x$ the $p_0$-percolation cluster that contains $x$.  Using a union bound,  
$$
\P[\diam( \mathcal C_{x}) =k ] \le
\sum_{\mathcal A\ni x,  |\mathcal A|= k} \P(\mathcal C_x \supseteq\mathcal A )
\le
Cp_0 e^{-Ak},
$$
 for some $A>0$,  where the sum is over all lattice animals that contain $x$.  In particular, by taking $p_0$ sufficiently small, we may assume $A>1 $.  This implies that  $\diam( \mathcal C_{x}) = \mathcal O_1(Cp_0)$.  Therefore $\diam( \mathcal C_{x})^{d+2} = \mathcal O_{\frac 1{d+2}}(Cp_0^{d+2})$.  We denote,  for every cube $\cu$ and $p\in\N$,
 $$
 \diam( \mathcal C(\cu))^p:= \sum_{x\in\cu} \diam( \mathcal C_{x}) ^p.
 $$
 Thus for any $z\in \cu_n$,  $\frac 1{|\cu_{n/2}|} \diam( \mathcal C(z+\cu_{n/2}))^{d+2} = \mathcal O_{\frac 1{d+2}}(Cp_0^{d+2})$. 

Denote the event 
\begin{align*} 
E_{n/2}:= \bigl\{ \text{ there is no  } p_0 \text{ cluster in } \cu_n  \text{ with diameter larger than }  3^{n/2}\bigr\}\,.
\end{align*} 
A similar union bound then yields $\P[E_{n/2}^c] \le e^{-\kappa 3^{n/2}}$ for some $\kappa>0$.  Moreover, given the event $E_{n/2}$,  let $z_1, z_2\in \cu_n$ such that $\dist(z_1+\cu_{n/2}, z_2+\cu_{n/2}) > 4\cdot3^{n/2}$,  then $\frac 1{|\cu_{n/2}|} \diam( \mathcal C(z_1+\cu_{n/2}))^{d+2}  $ and $\frac 1{|\cu_{n/2}|} \diam( \mathcal C(z_2+\cu_{n/2}))^{d+2}  $ are independent,  as the two random variables are measurable with respect to the edges in $z_1+2\cu_{n/2}$ and $z_2+2\cu_{n/2}$ respectively.

We then apply~\cite[Lemma A.11]{AKMbook} separately to $4^d$ checkerboards of $\{(z+\cu_{n/2}): z\in 3^{n/2} \Zd\cap \cu_{n+1}\}$ to conclude that given the event $E_{n/2}$,
 \begin{align*} 
\sum_{x\in [\mathcal C_{>} \cap ( z+ \cu_{n+1})]} (\diam ( \mathcal C_{>}(x))^{d+2}
 \le_{st}
\sum_{z\in 3^{n/2}\Z\cap \cu_{n+1}} (\diam ( \mathcal C(z+\cu_{n/2}))^{d+2}
\le
4^d| \cu_{n/2} | \bar { \mathcal O} _{\frac 1{d+2}} (C3^{nd/4} p_0^{\frac{d+2}2}).
 \end{align*} 
 This implies 
  \begin{align*} 
  \frac{\diam ( \mathcal C_{>} \cap  (\cu_{n}) )^{d+2} }{|\cu_n|}
  =
C3^{-nd/4}  \bar { \mathcal O} _{\frac 1{d+2}} ( p_0^{\frac{d+2}2})\,.
  \end{align*} 
  By choosing $p_0$ sufficiently small, this implies that given $E_{n/2}$,  \eqref{e.d+1moment} holds.  
  
Fix this $p_0$,  the same argument also yields the estimate \eqref{e.pmoment} for $\frac{\diam ( \mathcal C_{>} \cap  (\cu_{n}) )^{p} }{|\cu_n|}$.This concludes~\eqref{e.si}. 
\end{proof}

Using the inverse moments bound~\eqref{e.invmom},  we may establish a Poincar\'e inequality at large-scales.  For every $U\subset \Zd$,  define  $H^1_\a(U)$ the set of measurable functions  $w:U\to \R$ with respect to the norm
\begin{equation*}
\| w \|_{H^{1}_\a(U)}
:=
\Biggl(
\sum_{e\in \mathcal E(U)}
\nabla w(e) \cdot\a(e) \nabla w(e)
+
\sum_{x\in U} |w(x)|^2\Biggr)^{\!\sfrac12}
 \,.
\end{equation*}
We also define
\begin{equation*}
\| w \|_{H^{1}(U)}
:=
\Biggl(
\sum_{e\in \mathcal E(U)}
(\nabla w(e) )^2
+
\sum_{x\in U} |w(x)|^2\Biggr)^{\!\sfrac12}
 \,.
\end{equation*}
And $H_0^1(U)$ the space of functions $u\in H^{1}(U)$ which satisfy $u=0$ on $\partial U$. 

\begin{proposition}[{Large-scale Poincar\'e inequality.}]
\label{p.poincare}
There exists a random variable~$\Sp$ satisfying  
\begin{align}
\label{e.Sfinite}
\P \left[ \Sp < \infty \right] = 1,
\end{align}
and an absolute constant $C<\infty$ 
such that, for every triadic cube~$\cu$ with $0\in \cu$ and $\diam(\cu) \geq \Sp$, we have, for every~$u\in H^1_\a(3\cu)$, 
\begin{equation}
\label{e.ass.poincare}
\inf_{s\in\R}
\frac{1}{|\cu|}\sum_{x\in \cu} 
| u (x)- s|^2
\leq 
C
\diam(\cu)^2
\frac{1}{|\cu|} \sum_{e\in \mathcal E(\cu)} 
\left| \a^{\frac12} (e)\nabla u(e) \right|^{2}
\,.
\end{equation}
\end{proposition}
\begin{proof}
Define $2^* := \frac{2d}{d+2}$ for $d\ge 3$ and $2^* := 3/2$ for $d=2$.  
We observe that by applying the Sobolev-Poincar\'e inequality to any $u$, 
\begin{align*}
\inf_{s\in\R}
\frac{1}{|\cu|}\sum_{x\in \cu} 
| u (x)- s |^2
&
\leq 
C\diam(\cu)^2
\bigg(
\frac{1}{|\cu|}\sum_{e\in \mathcal E( \cu)} 
\left| \nabla u (e)\right|^{2^*}
\bigg)^{\frac{2}{2^*}} 
\\ &
\leq 
C \diam(\cu)^2
\bigg(
\frac{1}{|\cu|}\sum_{e\in \mathcal E( \cu)} 
\a^{-\frac {2^{*}}{2}}(e)
\left| \a^{\frac12}(e) \nabla u(e) \right|^{2^*}
\bigg)^{\frac{2}{2^*}}
\\ & 
\leq
C\diam(\cu)^2
\bigg(
\frac{1}{|\cu|}\sum_{e\in \mathcal E( \cu)} 
\a^{-\frac {2^*}{2-2^*}}(e)
\bigg)^{\frac 2 {2^*} -1}
\bigg(
\frac{1}{|\cu|}\sum_{e\in \mathcal E( \cu)} 
\left| \a^{\frac12}(e) \nabla u(e) \right|^{2}
\bigg)\,,
\end{align*}
where we applied the H\"older inequality with $q= \frac 2 {2^*}$ and $p = \frac 2 {2-2^*}$ to obtain the last line.  
Notice that,  if we have that the cube~$\cu$ past a minimal scale $\mathcal S$ in Definition~\ref{d.good},  then  by the ergodic theorem,  we have for every $p\ge 1$, 
\begin{align*}
\bigg(
\frac{1}{|\cu|}\sum_{e\in \mathcal E( \cu)} 
\a^{-p}(e)
\bigg)^{-\frac 1p}
\geq \frac12\left\langle 
\a^{-p}(e)
\right\rangle^{-\frac 1p}_{\mu_\tau},
\end{align*}
Apply the above line with $p=\frac {2^*}{2-2^*}$,  we have 
\begin{align*}
\inf_{s\in\R}
\frac{1}{|\cu|}\sum_{x\in \cu} 
| u (x)- s |^2
\leq 
\frac{C \diam(\cu)^2}{\Bigl \langle 
\a^{-\frac {2^*}{2-2^*}}(e)
\Bigr \rangle_{\mu_\tau}^{1-\frac 2{2^*}}}
\bigg(
\frac{1}{|\cu|}\sum_{e\in \mathcal E( \cu)} 
\left| \a^{\frac12}(e) \nabla u(e) \right|^{2}
\bigg)
\,.
\end{align*}
The conclusion follows from the bound on the inverse moments of $\a$ stated in Corollary~\ref{c.mgf}.
\end{proof}

It is clear that we can take $\S_P$ to be the minimal scale for the good cube, thus have the same stochastic integrability~\eqref{e.si}.

\subsubsection{Subadditive quantities}

We now define subadditive quantities with respect to the Gibbs measure $\mu_\tau$. For every subset $U \subset \Z^d$,  define the set of all edges in $U$ by
\begin{equation*}
\mathcal{E}(U)
:=
\left\{ 
(x,y) \in U\times U \,:\, 
x\sim y, \
x \ll y
\right\}.
\end{equation*}
For $f : \Z^d \to \R$,  we define 
\begin{equation}
\label{e.hat}
\hat f(x) := 
\left\{
\begin{aligned}
& f(x),  \quad \text{ if }  x\in \Z^d \setminus  \mathcal C_{>},
\\
& \frac 1{|\partial \mathcal C_{>}(x)|} \sum_{y \in \partial \mathcal C_{>}(x)} f(y), \quad \text{ if }  x\in  \mathcal C_{>}
\,.
\end{aligned}
\right.
\end{equation}
Namely,  $\hat f$ replaces the value of $f$ on each connected cluster of large conductances by the average of $f$ over the boundary of the cluster.  

\smallskip

We define, for every subset~$U$ and~$p,q\in\Rd$, 

\begin{equation}
\label{e.nu}
\nu(U,p) = \inf_{v\in \hat{\ell_p }+ H_0^1(U)} \frac{1}{|U|} \sum_{e\in \mathcal{E}(U)}
\frac12 \nabla v \cdot \a \nabla v 
\end{equation}
and 
\begin{equation}
\label{e.nustar}
\nu^*(U,q) = \sup_{v\in H^1(U)} \frac{1}{|U|} \sum_{e\in \mathcal{E}(U)}
\left(-\frac12 \nabla v \cdot \a \nabla v  + q\cdot \nabla v \right)
\,.
\end{equation}
It is clear that the infimum and supremum above are (uniquely) attained, by noticing that $\a(e,e) = e^{\tau_e}>0$ with probability one, which gives coercivity of the energy functionals with respect to $H^1(U)$. 

We devote the rest of this section to collecting some basic properties of the quantities $\nu$ and $\nu^*$.  We show first that~$\nu$ and $\nu^*$ are actually quadratic polynomials and compute the first and second variations of their defining optimization problems. The following lemma is analogous to~\cite[Lemma 2.2]{AKMbook} and~\cite[Lemma 4.1]{AKbook}.  Given $U\subset \Zd$,  define 
\begin{equation}
\A(U):= \bigl\{u\in H^1(U): -\nabla\cdot \a \nabla u=0 \text { in } U\bigr\}\,.
\end{equation}
We also define,  for $w \in  H^1(U)$,
\begin{equation*}
\mathsf{E}_{U}\left[ w \right] = \frac{1}{|U|} \sum_{e\in \mathcal{E}(U)}
\frac12 \nabla w \cdot \a \nabla w.
\end{equation*}

\begin{lemma}
[Properties of $\nu$ and $\nu^*$]
\label{l.basicprops}
Fix a cube~$U \subseteq \Zd$. The quantities $\nu(U,p)$ and $\nu^*(U,q)$ and their respective optimizing functions $v(\cdot,U,p)$ and $u(\cdot,U,q)$ satisfy the following:

\begin{itemize}

\item \emph{Quadratic representation.}
There exist symmetric matrices $\ahom(U), \ahom_*(U) \in \R^{d\times d}$ such that
\begin{equation}
\label{e.quadrep}
\left\{
\begin{aligned}
& \nu(U,p)
= \frac12 p\cdot \ahom(U) p  \quad \forall p\in\Rd, \quad \mbox{and}
\\ & 
\nu^*(U,q) 
= \frac 12  q \cdot \ahom_*^{\,-1}(U)q  \quad \forall q\in\Rd. 
\end{aligned}
\right.
\end{equation}

\item \emph{First variation.} The optimizing functions are characterized as follows: $v(\cdot,U,p)$ is the unique element of $\hat \ell_p+H^1_0(U)$ satisfying 
\begin{align}
\label{e.firstvar.nu}
\sum_{e\in \mathcal{E}(U)} 
 \nabla v(e,\cdot) \a (e,\cdot) \nabla w(e,\cdot)
=
0
, \quad \forall w \in H^1_0(U);
\end{align}
$u(\cdot,U,q)$ is the unique element of $H^1(U)$ satisfying $\sum_{e\in \mathcal{E}(U)} 
 u(e,\cdot) =0$ and
\begin{align}
\label{e.firstvar.nustar}
\frac1{|U|} \! \sum_{e\in \mathcal{E}(U)} 
\nabla u(e,\cdot) \a (e,\cdot) \nabla w(e,\cdot)
=
\frac1{|U|} \sum_{e\in \mathcal{E}(U)} 
\nabla \ell_q(e)
\nabla w(e,\cdot) ,
 \quad \forall w \in H^1(U).
\end{align}

\item \emph{Second variation.}
For every~$w\in\hat \ell _{p}+H_{0}^{1}\left( U\right)$, 
\begin{equation}
\label{e.quadresp.nu}
\mathsf{E}_{U}\left[ w \right] -  \nu(U,p)
= 
\frac1{|U|} \sum_{e\in \mathcal{E}(U)} 
\frac 12 \nabla (v- w) \a (e,\cdot) \nabla (v-w)
\end{equation}
and, for every $w\in H^1(U)$, 
\begin{align}
\label{e.quadresp.nustar}
 \nu^*(U,q)- 
\biggl ( 
\frac1{|U|}\! \sum_{e\in \mathcal{E}(U)} \!
\nabla \ell_q(e)
\nabla w(e,\cdot) 
- 
\mathsf{E}_{U}[w] 
\biggr )
=
\frac1{|U|} \!\sum_{e\in \mathcal{E}(U)} \!
\frac 12 \nabla (u- w) \a (e,\cdot) \nabla (u-w).
\end{align}

\item \emph{Spatial averages of gradients and fluxes.}  For every $p,q\in\R$,  
\begin{equation}
\label{e.spatial}
\ahom_*^{-1} (U) q = \ahom_*^{-1} (U) \frac1{|U|} \sum_{e\in \mathcal{E}(U)}  \a \nabla u (e,q)=\frac1{|U|} \sum_{e\in \mathcal{E}(U)}  \nabla u (e,q)
\end{equation}

\item\emph{Lower bound for the energy.}  For every $u\in H^1(U)$,
\begin{equation}
\label{e.gradupper}
\frac 12 \biggl ( \frac 1{|U|}\sum_{e\in\mathcal E(U)}\nabla u(e)\biggr ) \cdot \ahom_*(U) \biggl ( \frac 1{|U|}\sum_{e\in\mathcal E(U)}\nabla u(e)\biggr )
\leq
\frac1{|U|} \sum_{e\in \mathcal{E}(U)} 
\frac 12 \nabla u \cdot \a (e,\cdot) \nabla u.
\end{equation}
\end{itemize}
\end{lemma}

\begin{proof}

We fix a domain $U$ and $p,q\in\Rd$ and set $v:=v(\cdot,U,p)$, $u:=u(\cdot,U,q)$ to ease the notation.  
\smallskip

\emph{Step 1.} We prove the first and second variation formulas. Observe that, for~$t\in\R$ and $w\in H^1_0(U)$,  
\begin{align*}
\lefteqn{
 \frac{1}{|U|} \sum_{e\in \mathcal{E}(U)}
\frac12 \nabla (v+tw) \cdot \a \nabla (v +tw)
} \qquad & \\ &
= 
 \frac{1}{|U|} \sum_{e\in \mathcal{E}(U)}
\frac12 \nabla v \cdot \a \nabla v 
+
t \frac{1}{|U|} \sum_{e\in \mathcal{E}(U)}
 \nabla v \cdot \a \nabla w
 +
 t^2 \frac{1}{|U|} \sum_{e\in \mathcal{E}(U)}
\frac12 \nabla w \cdot \a \nabla w.
\end{align*}
Using that~$v$ is the minimizer,  we may divide by $|t|$ in the previous display and send~$t\to 0+$ as well as~$t\to 0-$ to find that the coefficient of $t$ in the previous display vanishes. This yields~\eqref{e.firstvar.nu} as well as 
\begin{equation*}
 \frac{1}{|U|} \sum_{e\in \mathcal{E}(U)}
\frac12 \nabla (v+tw) \cdot \a \nabla (v +tw)
= 
 \frac{1}{|U|} \sum_{e\in \mathcal{E}(U)}
\frac12 \nabla v \cdot \a \nabla v 
+
t^2 \frac{1}{|U|} \sum_{e\in \mathcal{E}(U)}
\frac12 \nabla w \cdot \a \nabla w.
\end{equation*}
As $\mathsf{E}_{U}\left[ w \right] = \frac{1}{|U|} \sum_{e\in \mathcal{E}(U)}
\frac12 \nabla w \cdot \a \nabla w$,  this also gives~\eqref{e.quadresp.nu}. The argument for~\eqref{e.firstvar.nustar} and~\eqref{e.quadresp.nustar} is similar and we omit it. 

\smallskip
\emph{Step 2.} We prove the quadratic representations. Observe that as a consequence of the first variations,  the map
\begin{equation}
\label{e.tildelinearmapping}
p\mapsto 
v(\cdot,U,p) \quad \mbox{and} \quad q\mapsto 
u(\cdot,U,q) \quad \mbox{are linear.}
\end{equation}
Using the definition~\eqref{e.nu} and~\eqref{e.nustar},  we deduce that $\nu(U,p)$ and $\nu^*(U,q)$ are quadratic forms in $p$ and $q$ respectively.   Namely,  there exist matrices $\ahom(U) $ and $\ahom_*(U)$ such that~\eqref{e.quadrep} holds.

\smallskip
\emph{Step 3.} We prove~\eqref{e.spatial}.  By testing the first variation~\eqref{e.firstvar.nustar} with $w = u(\cdot,  U, q)$,  and applying~\eqref{e.quadrep},   we conclude $\ahom_*^{-1} (U) q =\frac1{|U|} \sum_{e\in \mathcal{E}(U)}  \nabla u (e,q)$.  Applying~\eqref{e.firstvar.nustar} with $w = \ell_p$ yields $ q =  \frac1{|U|} \sum_{e\in \mathcal{E}(U)}  \a \nabla u (e,q)$.

\smallskip
\emph{Step 4.} We prove~\eqref{e.gradupper}.  By the definition of $\nu^*$,  we have that for every $u\in H^1(U)$,  if we set $p:= \frac 1{|U|} \sum_{e\in \mathcal E(U)} \nabla u(e)$ then we have 
\begin{multline*}
\frac 12 p\cdot \ahom_*(U) p = \sup_{q\in\Rd} \left(p\cdot q - \frac 12 q\cdot \ahom_*^{-1}(U)q \right)\\
= \sup_{q\in\Rd} \inf_{v\in H^1(U)}  \frac 1{|U|} \sum_{e\in \mathcal E(U)} \left( q\cdot (\nabla u -\nabla v) +\frac 12 \nabla v \cdot \a \nabla v \right)
\leq
\frac1{|U|} \sum_{e\in \mathcal{E}(U)} 
\frac 12 \nabla u \cdot \a (e,\cdot) \nabla u.
\end{multline*}
This completes the proof. 
\end{proof}

We show next that the energy quantities $\nu$ and $\nu^*$ are bounded above and below by quadratic functions,  if we go beyond a random scale. 

\begin{lemma}[Upper bound for $\nu$ and $\nu^*$]
\label{l.moment}
Let $\S$ be the minimal scale for good cubes in Definitoin \ref{d.good}.  Let $\cu\subset \Zd$ be a good cube, namely $\diam(\cu) > \S$.  Then there exists $C<\infty$,  such that for any $p,q\in\Rd$ , 
\begin{equation}
 \nu(\cu,p) \leq   C |p|^2
\end{equation}
and 
\begin{equation}
\nu^*(\cu,q) \leq   C|q|^2\,.
\end{equation}
\end{lemma}

\begin{proof}
We first show the upper bound for $\nu$.  Take $v = \hat\ell_p$ in the definition of $\nu$,  we have 
\begin{align*}
 \nu(\cu,p) &\leq   
 \frac{1}{|\cu|} \sum_{e\in \mathcal{E}(\cu)}
\frac12 \nabla  \hat\ell_p \cdot \a \nabla  \hat\ell_p
\\ &
\leq
 \frac{1}{|\cu|} \sum_{e\in \mathcal{E}(\cu) \setminus \mathcal C_>}
\frac12 \nabla  \ell_p \cdot \a \nabla  \ell_p
+ 
 \frac{1}{|\cu|} \sum_{e\in \mathcal{E}(\cu) \cap \partial \mathcal C_>}
\frac12 \nabla ( \hat\ell_p-\ell_p) \cdot \a \nabla ( \hat\ell_p-\ell_p)
\,.
\end{align*}
Note that for $e=(x,y)\in \mathcal{E}(\cu) \cap \partial \mathcal C_>$
\begin{equation}
\label{e.graddiff}
|\nabla \hat\ell_p(e)- \nabla \ell_p(e)|
\leq 
p \cdot \diam ( \mathcal C_>(x)).
\end{equation}
Using that the boundary of a set in $\Zd$ can be bounded by its volume,  we have $|\partial \mathcal C_>| \leq |  \mathcal C_>| \leq C_d \diam ( \mathcal C_>(x))^{d}$,  for every $x\in \partial \mathcal C_>$,  therefore
\begin{equation}
\label{e.energyub}
 \frac{1}{|\cu|} \sum_{e\in \mathcal{E}(\cu) \cap \partial \mathcal C_>}
\frac12 \nabla ( \hat\ell_p-\ell_p) \cdot \a \nabla ( \hat\ell_p-\ell_p)
\leq 
C_d |p|^2 \frac{\diam ( \mathcal C_>\cap \cu)^{d+2}}{|\cu|}.
\end{equation}
Using the definition of good cubes Definition~\ref{d.good},  we see that the above quantity is bounded by 
$\frac 1{100}C_d |p|^2$,  and the upper bound for $\nu$ follows.  

To show the upper bound for $\nu^*$, we notice that its first variation~\eqref{e.firstvar.nustar} implies 
\begin{equation}
\nu^*(\cu,q) = \frac 12 \frac1{|\cu|} \sum_{e\in \mathcal{E}(\cu)} 
\nabla \ell_q(e)
\nabla u (e,\cdot)
=\frac1{|\cu|} \sum_{e\in \mathcal{E}(\cu)} 
\frac 12 \nabla u(e,\cdot) \a (e,\cdot) \nabla u(e,\cdot)
\,.
\end{equation}
By Young's inequality, 
\begin{equation}
\frac1{|\cu|} \sum_{e\in \mathcal{E}(\cu)} 
\nabla \ell_q(e)
\nabla u (e,\cdot)
\leq 
\frac1{|\cu|} \sum_{e\in \mathcal{E}(\cu)} 
\Bigl  (  q \cdot \a^{-1} q +  \frac 14\nabla u(e,\cdot) \a (e,\cdot) \nabla u(e,\cdot)
\Bigr  )
\,.
\end{equation}
Therefore 
\begin{equation}
\nu^* (\cu,q) 
\leq 4q\cdot \Biggl( \frac1{|\cu|} \sum_{e\in \mathcal{E}(\cu)} \a^{-1} (e,\cdot) \Biggr) q
\,.
\end{equation}
By the definition of good cubes Definition~\ref{d.good},   $\frac1{|\cu|} \sum_{e\in \mathcal{E}(\cu)} \a^{-1} (e,\cdot) \leq 2\langle \a^{-1}(e,\cdot)\rangle _{\mu_\tau}\leq C_1$,  thus implies the upper bound for $\nu^*$.

\end{proof}

\begin{lemma}[Fenchel inequality]
\label{l.fenchel}
There exists $n_0\in \N$ and $\kappa>0$,  and a random minimal scale $\S_F$  such that for every $m\in\N$,  we have
\begin{align} 
\label{e.s1}
\P[\S_F > 3^{n_0+m}] \leq e^{-\kappa 3^{m/2}},
\end{align}
and that, for every~$n\in\N$ with $3^n > \S_F$,  and
for every~$p,q\in\Rd$, we have
\begin{equation} 
\label{e.fenchel}
\nu(\cu_n,p)+\nu^*(\cu_n,q) \geq p\cdot q - C  3^{-\frac12 n(d-2+ \kappa \frac{4-d}{2(d-1)})} |p||q| 
\,.
\end{equation}
\end{lemma}
\begin{proof}
It is clear that any function~$w\in \ell_p+H^1_0(\cu_n)$ satisfies, for every $q\in\Rd$, 
\begin{equation} 
\frac1{|\cu_n|} 
\sum_{e\in \mathcal{E}(\cu_n)} 
\nabla \ell_q(e)\nabla w(e) = p\cdot q. 
\end{equation}
The lemma thus follows by testing the definition of~$\nu^*(\cu_n,q)$ with the minimizer of~$\nu(\cu_n,p)$,  which yields 
\begin{align*} 
\label{e.stokesie}
\frac1{|\cu_n|} 
\sum_{e\in \mathcal{E}(\cu_n)} 
\nabla \ell_q(e)\nabla v(e) = p\cdot q + \frac1{|\cu_n|} 
\sum_{e\in \mathcal{E}(\cu_n)} 
\nabla \ell_q(e) (\nabla \hat{ \ell_p}(e)- \nabla \ell_p(e)) \\
= p\cdot q +q\cdot \frac1{|\cu_n|}  \sum_{e\in \mathcal{E}(\cu_n)} (\nabla \hat{ \ell_p}(e)- \nabla \ell_p(e)).
\end{align*}
Notice that for $x \in \cu$,  
\begin{equation*} 
|\hat \ell_p(x) - \ell_p(x) | \leq 2p \diam(\mathcal C_{>} \cap 3\cu) .
\end{equation*} 
Let $\mathcal S_1$ be the minimal scale as in Lemma~\ref{l.goodS},  and $3^{n_0} = \mathcal S_1$. 
For $n_0 < m< n$,  denote by $\{z+\cu_{m}: z\in 3^m\Zd\cap \cu_n,  (z+\cu_{m})\cap \partial \cu_n \neq \emptyset\}$ the collection of the subcubes of size $3^m$ that touches $\partial \cu_n$.  Notice that the number of the boundary cubes of the form $z+\cu_m$ is $(3^{n-m})^{d-1}$.  We apply the Stokes Theorem and the Poincar\'e inequality  to conclude for every $m \ge n_0$,
\begin{align*} 
\frac1{|\cu_n|}  \Biggl|\sum_{e\in \mathcal{E}(\cu_n)} (\nabla \hat{ \ell_p}(e)- \nabla \ell_p(e)) \Biggr|
&= \frac1{|\cu_n|}  \Biggl|\sum_{x\in \partial \cu_n} ( \hat{ \ell_p}(x)-  \ell_p(x))\Biggr|
\leq
 \frac1{|\cu_n|} \Biggl |  \sum_z \Biggl ( \sum_{x\in z+\cu_{m}}| \hat{ \ell_p}(x)-  \ell_p(x)|\Biggr ) \Biggr|
\\
&\leq
\frac1{|\cu_n|}  \sum_z \Biggl(C 3^{2m} \sum_{e\in \mathcal{E}(z+\cu_{m})}  (\nabla \hat{ \ell_p}(e)- \nabla \ell_p(e))^2 \Biggr)^\frac 12 
\\ &
\leq 
C_d \frac{3^m }{3^{nd}} 3^{(n-m)(d-1)} \left( 3^{md}|p| ^2+ \diam ( \mathcal C_{>} \cap ( z+ \cu_{m}) )^{d+2} |p|^2) \right)^\frac 12
\,.
\end{align*} 
By the definition of $\mathcal S_1$,  $\diam ( \mathcal C_{>} \cap ( z+ \cu_{m}) )^{d+2} \le C3^{md}$,  therefore 
\begin{align}
\label{e.gradhatave}
\frac1{|\cu_n|}  \Biggl|\sum_{e\in \mathcal{E}(\cu_n)} (\nabla \hat{ \ell_p}(e)- \nabla \ell_p(e)) \Biggr|
\le
C \frac{3^m }{3^{nd}} 3^{(n-m)(d-1)}  3^{md/2}|p| \,.
\end{align} 
 Let $\S_F$ be the minimal scale at which all the cubes  in  $\{z+\cu_{m}: z\in 3^m\Zd\cap \cu_n,  (z+\cu_{m})\cap \partial \cu_n \neq \emptyset\}$ satisfy \eqref{e.d+1moment}. By the stochastic integrability \eqref{e.si},   we have 
$$
\P[\S_F >3^n]\leq  3^{-(m-n_0)\kappa} (3^{n-m})^{d-1}\,.
$$
Take $m = n(1- \frac\kappa{2(d-1)})\ge \frac n2 >n_0$,  we have 
\begin{align}
\label{e.S1}
\P[\S_F >3^n]\leq  3^{(-(\frac 12 - \frac \kappa{2(d-1)} )n +n_0)\kappa}\,,
\end{align}
which is summable in $n$. This implies $\S_F <\infty$ a.s.,  and for $3^n >\S_F$,  we use~\eqref{e.gradhatave} to conclude 
\begin{align} 
\frac1{|\cu_n|}  \Biggl|\sum_{e\in \mathcal{E}(\cu_n)} (\nabla \hat{ \ell_p}(e)- \nabla \ell_p(e))\Biggr|
\leq 
C_d \frac{3^m }{3^{nd}} 3^{(n-m)(d-1)} 3^{md/2}|p|
\leq C |p| 3^{-\frac12 n(d-2+ \kappa \frac{4-d}{2(d-1)})}
\label{e.gradavewiggle}
\end{align} 
and therefore the conclusion follows.  
\end{proof}

Combining the previous two lemmas,  we conclude the lower bound for $\nu$ and $\nu^*$ beyond some random scale. 
\begin{corollary}[Lower bound for $\nu$ and $\nu^*$]
\label{c.lowerbd}
There exists a random minimal scale $\S$ with $\P[S<\infty]=1$ and a  constant $\kappa>0$ such that the following statement holds.
For $3^n > \S$, let $\cu_n\subseteq\Zd$ be a cube.   Then,  
for every~$p,q\in\Rd$, 
\begin{equation*} 
 \nu(\cu_n,p) \geq c_1 |p|^2 -C  3^{-\frac12 n(d-2+ \kappa \frac{4-d}{2(d-1)})} |p|^2,
 \end{equation*} 
 and
 \begin{equation*} 
 \nu^*(\cu_n, q) \geq c_1 |q|^2 - C  3^{-\frac12 n(d-2+ \kappa \frac{4-d}{2(d-1)})}|q|^2.
 \end{equation*} 
\end{corollary}

We next show that the quantity $\nu$ and $\nu^*$ are approximately subadditive.    

\begin{lemma}[{Subadditivity of $\nu$ and $\nu^*$}]
\label{l.nussubadd}
There exists an a.s. finite random minimal scale $\S$  such that,  let $U_1, \cdots U_n$ be disjoint triadic cubes with $\diam(U_i) \ge \S$ for each $i = 1,\cdots, n$,  and $U = \bigcup_{i=1}^n U_i$.  Let $\mathcal{E}': = \mathcal{E}(U) \setminus \bigcup_{i=1}^n \mathcal{E}(U_i)$. For every $p\in\Rd$, 
\begin{equation}
\label{e.nu.subadd}
\nu(U, p) 
\leq 
\sum_{i=1}^n \frac{|U_i|}{|U|} \nu(U_i, p) \,.
\end{equation}
Moreover, 
for every~$q\in\Rd$, we have
\begin{equation}
\label{e.nus.subadd}
\sum_{i=1}^n \frac{|U_i| }{|U|}\nu^*(U_i, q)
\geq 
\nu^*(U,q)
-\frac{C| \mathcal E' |^\frac14}{|U|^\frac14}   |q|^2 .
\end{equation}
\end{lemma}
\begin{proof}
Define a function $v_g \in \hat\ell_p+H^1_0(U)$  by 
\begin{equation*} \label{}
v_g(x):= v(x,U_i,p),\quad \text{if }x\in U_i .
\end{equation*}
Testing the definition of $\nu(U,p) $ with $v_g$ gives 
\begin{align*}
\nu(U,p) &\le  \frac{1}{|U|} \sum_{e\in \mathcal{E}(U)}
\frac12 \nabla  v_g \cdot \a \nabla  v_g   \\
&= \sum_{i=1}^n \frac{1}{|U|} \sum_{e\in \mathcal{E}(U_i)}\frac12 \nabla  v_g \cdot \a \nabla  v_g
- \frac{1}{|U|} \sum_{e\in \mathcal E'} \frac12 \nabla \hat \ell_p(e) \cdot \a \nabla \hat \ell_p(e) \leq
\sum_{i=1}^n \frac{|U_i|}{|U|} \nu(U_i, p) .
\end{align*}

We next show the subadditivity of $\nu^*$.  Testing the definition of $\nu^*(U_i,q)$ with the function $u(\cdot,  U,q)$ and summing over $i=1, \cdots, n$ yields,
\begin{align*}
\sum_{i=1}^n  |U_i| \nu^*(U_i,  q)
&\geq 
\sum_{i=1}^n \sum_{e\in \mathcal E(U_i)} \Bigl( -\frac 12 \nabla u \cdot \a\nabla u + q\cdot \nabla u   \Bigr)\\
&= |U| \nu^* (U, q) - \sum_{e\in \mathcal E'} \Bigl( -\frac 12 \nabla u \cdot \a\nabla u + q\cdot \nabla u    \Bigr)
\geq
|U| \nu^* (U,q) - \sum_{e\in \mathcal E'} q\cdot \nabla u .
\end{align*}
Let $\S$ be the minimal random scale in Definition \ref{d.good}.  We therefore have, by the H\"older inequality, 
\begin{align*}
\sum_{e\in \mathcal E'} q \cdot \nabla u
&\leq 
\Biggl(  \sum_{e\in \mathcal E' } |q|^4 \Biggr)^{\!\frac 14}
\Biggl( \sum_{e\in\mathcal E(U) } \a(e)( \nabla u(e))^2 \Biggr)^{\!\frac 12}
\Biggl(\sum_{e\in \mathcal E'  } \a^{-2}(e)\Biggr)^{\!\frac 14} \\
&\leq
C|q|| \mathcal E' |^{\!\frac 14} \Biggl( |U|  \nu^* (U,q) \Biggr)^{\!\frac 12}
\Biggl( \sum_{i=1}^n\sum_{e\in \mathcal E(U_i)  } \a^{-2}(e)\Biggr)^{\!\frac 14} .
\end{align*}
Since $\diam(U_i)> \S$,  $(U_i)$  are good cubes as in Definition~\ref{d.good},  we have 
\begin{equation*} 
 \sum_{i=1}^n\sum_{e\in \mathcal E(U_i)  } \a^{-2}(e)
\leq
C \sum_{i=1}^n |U_i| \bigl\langle \a^{-2}(e) \bigr\rangle_{\mu_\tau} 
\\
\leq 
C|U|  \bigl\langle \a^{-2}(e) \bigr\rangle_{\mu_\tau} ,
\end{equation*}
and therefore 
\begin{align*}
\sum_{i=1}^n \frac{|U_i| }{|U|}\nu^*(U_i,q)
\geq 
\nu^*(U,q)
-\frac{C | \mathcal E' |^{\frac 14}|U|^\frac 12 |U|^\frac 14 }{|U|}  \langle \a^{-2}(e) \rangle_{\mu_\tau}^\frac 14 |q|^2
\geq 
\nu^*(U,q)
-\frac{C'| \mathcal E' |^\frac 14}{|U|^\frac 14}   |q|^2\,.
\end{align*}
The proof is complete.
\end{proof}

\subsection{Proof of homogenization}

Based on the properties of the subadditive quantities in the previous section,  we use a variational argument to prove qualitative homogenization for the equation 
\begin{equation*}
-
\nabla \cdot \a\nabla u = \nabla \cdot \mathbf{f} \quad \mbox{in} \ \Zd\,, 
\end{equation*}
where~$\f$ is a smooth, macroscopic function,
to the constant-coefficient, homogenized equation 
\begin{equation*}
-
\nabla \cdot \ahom\nabla \bar u = \nabla \cdot \mathbf{f}
\quad \mbox{in} \ \Rd\,.
\end{equation*}
Together with moments bound that follows from the Brascamp-Lieb inequality Proposition~\ref{p.BL},  we will conclude the  converge in law of the continuous SOS field to a GFF. 

\subsubsection{The subadditive ergodic theorem}

We recall a version of the multiparameter Kingman's subadditive ergodic theorem,  obtained in~\cite{AkKr}.  Denote by $\A:=\{ z+\cu_n:  z\in\Zd, n\in\N\}$ the collection of triadic cubes in $\Zd$. 

\begin{proposition}[Ergodic theorem]
\label{p.ergodic}
Let $(F_I)_{I\in\A}$ be a family of real-valued random variables satisfying:
\begin{itemize}
\item $F_I \circ \tau_z = F_{z+I}$ \quad \text{for every } $z\in\Zd, I\in \A$.

\item  If $I_1, I_2, \cdots I_n$ are disjoint sets in $\A$ and $I = \bigcup_{i=1}^n I_i$,  then 
$F_I \leq \sum_{i=1}^n \frac{|I_i|}{|I|} F_{I_i}$.

\item  $\inf_{I\in\A} \E F_I >-\infty$,
\end{itemize}
Then $\lim_{|I|\to \infty} F_I$ exists a.e. 
\end{proposition}
Recall that for $z\in\Zd$, the random variable $\S(z)$ denotes the minimal scale above which the cube centered at $z$ is a good cube as in Definition~\ref{d.good}.  Notice that,  if we take $U = \cu_n$ and $U_1, \cdots,  U_{3^d}$ be one if its subcube of size $3^{n-1}$,~\eqref{e.nu.subadd} and~\eqref{e.nus.subadd} becomes 
\begin{equation*}
\nu(U,p) 
\leq 
\sum_{i=1}^n \frac{|U_i|}{|U|} \nu(U_i, p)  .
\end{equation*}
and 
\begin{equation*}
\sum_{i=1}^n \frac{|U_i| }{|U|}\nu^*(U_i, q)
\geq 
\nu^*(U,q)
-\frac{C}{|U|^\frac{1}{4d}}   |q| ^2 .
\end{equation*}
We apply Proposition~\ref{p.ergodic} above with
$$
F_{z+I} = \frac 1{|p|^2} \nu(z+I, p)\indc_{\{ \text{diam}(z+I) \ge \S_1(z)\}} + \infty \indc_{\{ \text{diam}(z+I) < \S_1(z)\}},
$$
and
$$F_{z+I}^* = \frac 1{|q|^2} \nu^*(z+I, q)\indc_{\{\text{diam}(z+I) \ge \S_1(z)\}} + \infty \indc_{\{\text{diam}(z+I) < \S_1(z)\}}+ C_{d} |z+I|^{-\frac{1}{4d}}.$$
 The conditions of $(F_I)_{I\in\A}$ and $(F_I^*)_{I\in\A}$ follows from the stationarity of $\nu$ and $\nu^*$,  the subadditivity Lemma~\ref{l.nussubadd} and the lower bound Corollary~\ref{c.lowerbd}.    Thus as $n\to \infty$, the subadditive quantity $\nu(\cu_n, q)$ and $\nu^*(\cu_n, q)$ convergences for every $p, q\in\Rd$ to some $\overline{\nu}(p)$ and $\overline{\nu}^*(q)$. Define,  for some positive definite matrices $\ahom$ and  $\ahom_* $ ,    
 \begin{equation}
\label{e.quadrep.homs}
\overline{\nu}(p) = \frac12p\cdot \ahom p ,
\end{equation}
and 
  \begin{equation}
\overline{\nu}^*(q) = \frac12 q\cdot \ahom_*^{-1} q ,
\end{equation}
such that 
 \begin{equation}
 \label{e.conv.bars}
| \ahom(\cu_n) -\ahom| + | \ahom_*(\cu_n) -\ahom_*| \to 0 ,  \quad \text{$\P$--a.s.}
 \end{equation}
 
 We will show below that $p \mapsto \bar \nu(p) $ and $q \mapsto \bar \nu^*(q) $ are convex dual functions,  and the two sets of the limiting coefficient are equal: $\ahom = \ahom_*$.  
 
 Our proof of homogenization will be based on approximating the limiting energy quantity by simplicial domains.  Below we introduce some notations.  We denote the set of permutations of~$\{ 1,\ldots, d\}$ by~$\mathcal{P}$. 
Given~$\pi \in \mathcal{P}$, we define the simplex~$\triangle_0^\pi$ by
\begin{equation*}
\triangle_0^\pi := \Bigl\{ (x_1,\ldots,x_d) \in\Rd \,:\, -\frac12 < x_{\pi(1)} < x_{\pi(2)} < \cdots < x_{\pi(d)} < \frac12  \Bigr\} 
\,.
\end{equation*}
We scale and translate these by defining, for every~$n\in\Z$ and~$z \in\Rd$, 
\begin{equation*}
\triangle_n^\pi(z) := z+3^n\triangle_0^\pi\,.
\end{equation*}
Just as for cubes, each simplex is Lipschitz domains and the set of simplexes have the nice property that larger simplexes are partitioned by smaller ones. Indeed, a triadic simplex~$\triangle_m^\sigma(z)$ can be written as the disjoint union (up to a set of Lebesgue measure zero) of~$3^{d(m-n)}$ many simplexes of the form~$\triangle_n^\omega(z)$. Moreover, a triadic cube~$\cu_n(z)$ is the disjoint union of~$\{ \triangle_n^\sigma(z) \,:\, \sigma\in \mathcal{P} \}$. Note that, since~$|\mathcal{P}| = d!$, we deduce that~$| \triangle_n^\pi| = 3^n/d!$.  We denote by~$S_n := \{ \triangle_n^\pi(z)  \,:\, n\in\Z\,,\, z\in 3^n\Zd\,, \, \triangle_n^\pi(z)  \subset \cu_m\}$ the collection of \emph{triadic adapted simplexes} in~$\cu_m$. 
We call~$3^n$ the \emph{size} of~$\triangle_n^\pi(z)$. 
If~$\pi$ is the trivial permutation, then we write~$\triangle_n = \triangle_n^\pi$ . 
Observe that a simplex has exactly~$d+1$ many extreme points; we call these the \emph{vertices} of the simplex.

\begin{lemma}
\label{l.simpconv}
Let $\S$ be the minimum scale for good cubes in Definition \ref{d.good}. 
For every $p,q\in\Rd$, $3^n> 2\S$ and $\pi \in \mathcal{P}$ ,
\begin{equation}
\label{e.convergence}
\left|\nu(\triangle_n^\pi,p)  - \overline{\nu}(p) \right| 
+ \left| \nu^*(\triangle_n^\pi,q) - \overline{\nu}^*(q) \right| 
\to 0. 
\end{equation}
\end{lemma}

\begin{proof}
The convergence of $\nu$ and $\nu^*$ quantity of the simplexes follows from the convergence along the triadic cubes,  and the subadditivity Lemma~\ref{l.nussubadd}.  We focus on the case of $\nu$ and the proof of $\nu^*$ is almost identical. We claim,  separately, that
\begin{equation*}
\nu(\triangle_n^\pi,p)  \le \overline{\nu}(p)+ o_n(1) ,
\end{equation*}
and 
\begin{equation*}
\sum_{\pi \in \mathcal{P}}\nu(\triangle_n^\pi,p)  \ge d! \overline{\nu}(p)- o_n(1) ,
\end{equation*}
from which we conclude $\nu(\triangle_n^\pi,p)  \to \overline{\nu}(p)$.  

\smallskip

To prove the upper bound,  take $n$ such that $3^n> 2\S$,  we tile each $\triangle_n^\pi$ with triadic  cubes and apply Lemma~\ref{l.nussubadd}.   Specifically,  it suffices to consider $\pi$ being the identical permutation (the other cases can be done by permutation symmetry).  Let $n_d$ be the largest integer such that $d 3^{n_d} \le 3^n$.  We tile $\triangle_n^{Id}$ using triadic cubes of size $3^{n_d},  3^{n_{d-1}}, \cdots ,  3^{n/2}$ following a Whitney type decomposition, which we now define. 

\smallskip

For any bounded Lipschitz domain $U \subset \Rd$,  we recursively define 
\begin{equation}
\left\{
\begin{aligned}
& T_0:
= \{ z+ \cu_0: \dist(z, \partial U ) \ge 3, z\in \Zd \cap U\}
\\ & 
\tilde U_0
:= T_0.
\end{aligned}
\right.
\end{equation}
And, having defined $T_0, \cdots,  T_j$ we set 
\begin{equation}
\left\{
\begin{aligned}
& T_{j+1}:
= \{ z+ \cu_{-(j+1)}: \dist(z, \partial U ) \ge 3^{-j}, z\in 3^{-j-1} \Zd \cap (U\setminus \tilde U_j)\}
\\ & 
\tilde U_{j+1}
:= \tilde U_{j} \cap T_{j+1}.
\end{aligned}
\right.
\end{equation}
Then $U$ can be written as disjoint union of $\{ z+ \cu_{-j}: z+ \cu_{-j} \in T_j \}$, $j\in \N$.  We also note that 

$$
\tilde U_{j+1} \setminus \tilde U_{j}  \subset
\{ x\in U: \dist(x, \partial U ) \in [C3^{-j}, C3^{-j+1}] \},
$$
thus $| \tilde U_{j+1} \setminus \tilde U_{j} | \le C_U 3^{-j}$.  

 Define the triadic cubes $\T_{n_d},  \cdots,  \T_{n/2}$ as rescaled version of $T_j$, namely,  $\T_{n_d-j}:= 3^{n_d} T_j$,  for $j = 0, \cdots,  n_d- n/2$.    Notice that there exists a constant $C<\infty$, such that for each $m = n_d,  \cdots,  n/2$,   $|\partial \T_m| \le C 3^{n(d-1)}$.  Moreover,  the volume of the boundary layer which is not covered by the tiles is bounded by 
\begin{equation*}
|\mathcal B_n| := \left|\triangle_n^{Id} \setminus \bigcup_{m=n/2}^{n_d} \T_m \right|
\leq
C3^{n(d-1)+n/2}.
\end{equation*}
Apply Lemma \ref{l.nussubadd},  we conclude that 
\begin{equation*}
\nu(\triangle_n^\pi,p)  \le
\sum_{m=n/2}^{n_d} \sum_{z+\cu_m \in \T_m} \frac{|z+\cu_m|}{|\triangle_n^\pi|} \nu(z+\cu_m,p) + \frac {|\mathcal B_n| }{|\triangle_n^\pi|}  \nu(\mathcal B_n,p).
\end{equation*}
Since $|\mathcal{E}'|  \le \sum_{m=n/2}^{n_d} |\partial \T_m|  \leq C n 3^{n(d-1)}$,  the above inequality can be simplied as 
\begin{equation}
\label{e.upptriangle}
\nu(\triangle_n^\pi,p)  \le
\sum_{m=n/2}^{n_d} \sum_{z+\cu_m \in \T_m} C3^{(m-n)d} \nu(z+\cu_m,p) + 3^{-n/2}  \nu(\mathcal B_n,p).
\end{equation}
Using the convergence of the energy along the triadic cubes,  for each $m = n_d,  \cdots,  n/2$, we have that $\nu(z+\cu_m,p)  \le \overline{\nu}(p)+o_n(1)$.  Moreover,
by plugging in  $ \hat\ell_p$ in the definition of $\nu(\mathcal B_n,p)$,  we have 
\begin{align*}
\nu(\mathcal B_n,p) &\leq   
 \frac{1}{|\mathcal B_n|} \sum_{e\in \mathcal{E}(\mathcal B_n)}
\frac12 \nabla  \hat\ell_p \cdot \a \nabla  \hat\ell_p 
\\ &
\leq
 \frac{1}{|\mathcal B_n|} \sum_{e\in \mathcal{E}(\mathcal B_n) \setminus \mathcal C_>}
\frac12 \nabla  \ell_p \cdot \a \nabla  \ell_p
+ 
 \frac{1}{|\mathcal B_n|} \sum_{e\in \mathcal{E}(\mathcal B_n) \cap \partial \mathcal C_>}
\frac12 \nabla ( \hat\ell_p-\ell_p) \cdot \a \nabla ( \hat\ell_p-\ell_p) 
\,.
\end{align*}
Note that for $e=(x,y)\in \mathcal{E}(\mathcal B_n) \cap \partial \mathcal C_>$
\begin{equation}
\label{e.graddiff.2}
|\nabla \hat\ell_p(e)- \nabla \ell_p(e)|
\leq 
p \cdot \diam ( \mathcal C_>(x)).
\end{equation}
Controlling the boundary using the volume,  which gives $|\partial \mathcal C_>| \leq |\mathcal C_>| \leq C_d \diam ( \mathcal C_>(x))^{d}$,  for every $x\in \partial \mathcal C_>$,  we have 
\begin{equation}
 \frac{1}{|\mathcal B_n|} \sum_{e\in \mathcal{E}(\mathcal B_n) \cap \partial \mathcal C_>}
\frac12 \nabla ( \hat\ell_p-\ell_p) \cdot \a \nabla ( \hat\ell_p-\ell_p) \\
\leq 
C_d |p|^2 \frac{\diam ( \mathcal C_>\cap \mathcal B_n)^{d+2}}{|\mathcal B_n|}  .
\end{equation}
We may cover $ \mathcal B_n$ by the cubes of the form $(z+\cu_{n/2})$,  this yields 
$\diam ( \mathcal C_>\cap \mathcal B_n)^{d+2} \leq \sum_z \diam ( \mathcal C_>\cap (z+\cu_n))^{d+2}.$
Using the definition of good cubes Definition~\ref{d.good},  we see that the above quantity is bounded by $C |\mathcal B_n|$,  which implies that $\nu(\mathcal B_n,p)$ is bounded from above.  
Substitute these estimates into~\eqref{e.upptriangle} we conclude $\nu(\triangle_n^\pi,p)  \le \overline{\nu}(p)+ o_n(1) $. 

\smallskip

To obtain the lower bound,  we apply Lemma~\ref{l.nussubadd} with $U= \cu_n$ and $U_\pi = \triangle_n^\pi$ and noticing that $(\triangle_n^\pi)_{\pi \in \mathcal P}$ are disjoint to conclude 
\begin{equation*}
\nu(\cu_n, p) \le 
\frac 1{d!} \sum_{\pi \in \mathcal P} \nu(\triangle_n^\pi,p),
\end{equation*}
which concludes the lower bound. 
\end{proof}

\subsubsection{Estimates on finite-volume correctors}

We now prove that the solution to the Neumann problem is approximately affine on large scales. Given a cube $Q\subset \Zd$,  we denote as $u(\cdot, Q, p)$ the solution of a Neumann boundary condition: 
  \begin{equation}
\label{e.BVP.v}
\left\{ 
\begin{aligned}
&  \nabla \cdot \a \nabla  u(\cdot,Q,q) 
 =0
& \mbox{in} & \ Q^\circ , 
\\ & 
\a \nabla u(\cdot,Q,q) - \ell_q = 0& \mbox{on} & \ \partial Q.
\end{aligned}
\right.
\end{equation}
For $m<n$, we also denote by $\mathcal{Z}_{m}=3^{m}\mathbb{Z}^{d}\cap \cu _{n}$.  
\begin{lemma}
\label{l.corrector}
There exist~$\alpha (d),\kappa(d)>0$, $n_0(d)\in\N$,~$C(d)<\infty$ and a random variable~$\S'$ satisfying
\begin{equation*}
\P \bigl[ \S'<\infty \bigr]
=1,
\end{equation*}
such that, for every $n>2n_0$, 
\begin{equation}
\label{e.H-1u}
3^{-(n+1)} \left\|\nabla u(\cdot, \cu_n,  \ahom_*(\cu_n)p) - \nabla \hat \ell_p \right\|_{\underline{H}^{-1}(\cu_{n+1})}
\leq
C|p| \sum_{m=n_0}^{n}3^{m-n} |\ahom_*(\cu_m) -\ahom_*|^\frac 12 
 + C|p| 3^{-\alpha n}
 \,.
\end{equation}
Consequently, 
\begin{equation}
\label{e.corrasconv}
\P \biggl[ 
\lim_{n\to \infty} 
3^{-n} \left\|\nabla u(\cdot, \cu_n,  \ahom_*(\cu_n)p) -  \nabla \hat \ell_p  \right\|_{\underline{H}^{-1}(\cu_{n+1})}
=
0 \biggr] = 1\,.
\end{equation}
\end{lemma}
\begin{proof}
We first notice, that by~\eqref{e.spatial},  
 \begin{equation}
 \label{e.gradave}
(\nabla u(\cdot, \cu_m,  \ahom_*(\cu_n)p))_{\cu_m} =  \ahom_*(\cu_n)  \ahom_*^{-1}(\cu_m)p.
\end{equation}
We apply the multiscale Poincar\'e inequality (Proposition~\ref{p.MP}) to obtain

\begin{align}
\label{e.gradMP}
3^{-n} \left\|  \nabla u(\cdot, \cu_n,\ahom_*(\cu_n) p) -  \nabla \hat \ell_p  \right\|_{\underline{H}^{-1}(\cu_{n})}
 \leq
C3^{-n} \left(   \nabla u(\cdot, \cu_n, \ahom_*(\cu_n)p) -   \nabla \hat \ell_p \right)_{\cu_{n}} \notag \\
+C\sum_{m=0}^{n-1}3^{m-n} \Biggl( 
\frac1{\left|\mathcal{Z}_{m}\right|}
\sum_{y\in \mathcal{Z}_{m}}
 \left| \left(  \nabla u(\cdot, y+\cu_n, \ahom_*(\cu_n)p) -   \nabla \hat \ell_p \right) _{y+\cu _{m}}\right|^{2}  \Biggr)^{\!\sfrac12}. 
 \end{align}
 We apply the triangle inequality to estimate the second term of the right side:
\begin{align*}
\lefteqn{ \frac1{\left|\mathcal{Z}_{m}\right|}
\sum_{y\in \mathcal{Z}_{m}}
 \left| \left(  \nabla u(\cdot, y+\cu_n, \ahom_*(\cu_n)p) -   \nabla \hat \ell_p \right) _{y+\cu _{m}}\right|^{2}
} \qquad &
\notag \\ &   
\leq 
\frac1{\left|\mathcal{Z}_{m}\right|}
\sum_{y\in \mathcal{Z}_{m}}
 \left| \left(  \nabla u(\cdot, y+\cu_m, \ahom_*(\cu_n)p) -  \nabla \hat \ell_p \right) _{y+\cu _{m}}\right|^{2}
\\ & \qquad 
+
\frac1{\left|\mathcal{Z}_{m}\right|}
\sum_{y\in \mathcal{Z}_{m}}
 \left| \left(\nabla u(\cdot, y+\cu_n, \ahom_*(\cu_n)p) -\nabla u(\cdot, y+\cu_m, \ahom_*(\cu_n)p)\right) _{y+\cu _{m}}\right|^{2}  \,.
\end{align*}
  By~\eqref{e.gradave},
  \begin{equation}
  \label{e.firsttermMP}
\bigl| \bigl(  \nabla u(\cdot, y+\cu_m, \ahom_*(\cu_n)p) -  \nabla \hat \ell_p \bigr) _{y+\cu _{m}}\bigr|^{2}
  \leq 
  \left|  \ahom_*(\cu_n)  \ahom_*^{-1}(\cu_m)p - p\right|^2
  +  \bigl| \bigl( -  \nabla \hat \ell_p  +  \nabla  \ell_p\bigr) _{y+\cu _{m}}\bigr|^{2}\,.
  \end{equation}
We let~$\S_F$ be the minimal scale defined in Lemma~\ref{l.fenchel}. Applying~\eqref{e.gradavewiggle},  we conclude that there exists $C(d) <\infty$ and $\kappa(d)>0$, such that for every $3^m\ge  \S_F$,  
  $$ \bigl|\bigl( -  \nabla \hat \ell_p  +  \nabla  \ell_p\bigr) _{y+\cu _{m}}\bigr|^{2}
  \le
  C |p|^2 3^{- m(d-2+ \kappa \frac{d+2}{4(d-1)})} \,.
  $$
By taking $m=n$ in~\eqref{e.firsttermMP},  we also have, for every~$3^n\geq \S_F$, 
   \begin{equation}
   \label{e.graduave}
\bigl|\bigl( \nabla u(\cdot, \cu_n, \ahom_*(\cu_n)p) -  \nabla \hat \ell_p \bigr) _{\cu _{n}} \bigr|
  \leq 
 C |p| 3^{- \frac 12 n(d-2+ \kappa \frac{d+2}{4(d-1)})}\,.
\end{equation}
We now bound the left side of \eqref{e.firsttermMP} for $3^m < \S_F$.  Define the random variable
\begin{align}
\label{e.S2}
\S_2 := \min\{ n\in\N:  \text{ there is no  } \mathcal C_> \text{ cluster in } \cu_n  \text{ with diameter larger than }  3^{n/4}\} \,.
 \end{align}
By the percolation estimate~\eqref{e.verysub} and a union bound,  there exists~$C(d)<\infty$ and~$\kappa(d)>0$, such that, for every $t \ge 1$, 
\begin{equation} 
\P[\S_2 > Ct] \leq Ce^{-\kappa \sqrt t} 
\,.
\end{equation}
In particular, $\S_2$ is almost surely finite and
\begin{equation} 
\label{e.si2}
\P[\S_2 > C3^m] \leq Ce^{-\kappa 3^{m/2}}.
\end{equation}
Let $\S':= \max\{ \S_F, \S_2\}$. 
Using~\eqref{e.firsttermMP} and the definition of $\S_2$ in~\eqref{e.S2},  we obtain, for every~$3^n \ge \S'$, and $m \in \N$ with~$3^m\leq \min\{ 3^n, \S_F \vee \S_2\}$,  we have 
\begin{equation*}
 \bigl| \bigl( -  \nabla \hat \ell_p  +  \nabla  \ell_p\bigr) _{y+\cu _{m}}\bigr|^{2}
 \le
 |p|^2 \max_{x\in \cu_m} \diam(\mathcal C_>(x))^2
 \le
 3^{n/2}|p|^2,
\end{equation*}
which implies
 \begin{equation}
\bigl| \bigl(  \nabla u(\cdot, y+\cu_m, \ahom_*(\cu_n)p) -  \nabla \hat \ell_p \bigr) _{y+\cu _{m}}\bigr|^{2}
\leq 
C 3^{n/2}|p|^2.
\end{equation}
By the second variation~\eqref{e.quadresp.nustar} and~\eqref{e.gradupper},  
  \begin{align*}
\lefteqn{ 
\left| \left(\nabla u(\cdot, y+ \cu_n, \ahom_*(\cu_n)p) -\nabla u(\cdot, y+ \cu_m, \ahom_*(\cu_n)p)\right) _{y+\cu _{m}}\right|^{2}
} \qquad & 
\notag 
\\
&
\leq 
C\ahom_*^{-1}(\cu_m) 
\frac1{|\cu_m|} \sum_{e\in \mathcal{E}(\cu_m)} 
\frac 12  (\nabla u(\cdot,  y+\cu_n, \ahom_*(\cu_n)p) -\nabla u(\cdot, y+\cu_m, \ahom_*(\cu_n)p) )^2 \cdot \a (e,\cdot)\\
&\leq
C \ahom_*^{-1}(\cu_m)  (\nu^*(\cu_m,\ahom_*(\cu_n)p) - \nu^*(\cu_n,\ahom_*(\cu_n)p)) \\
&\leq
2C \ahom_*^{-1}(\cu_m)  |\ahom_*(\cu_n)p|^2 |\ahom_*(\cu_m) -\ahom_*|\,.
   \end{align*}
   Substitute the above estimates into~\eqref{e.gradMP},  summing over $m$ and using the fact that $ \ahom_*(\cu_m)$ is uniformly bounded for all $m> n_0$,  we conclude that there exists $\alpha = \alpha (d) >0$ and $\kappa = \kappa (d) >0$,  such that 
\begin{align*}
\lefteqn{
3^{-n} \left\|  \nabla u(\cdot, \cu_n,\ahom_*(\cu_n) p) -  \nabla \hat \ell_p  \right\|_{\underline{H}^{-1}(\cu_{n})} 
} \qquad &
\notag \\ &
\leq
C3^{-n}|p| +
C\sum_{m=0}^{n_0}3^{m-n} \biggl( 
\frac1{\left|\mathcal{Z}_{m}\right|}
\sum_{y\in \mathcal{Z}_{m}}
 \left| \left(  \nabla u(\cdot, y+\cu_n, \ahom_*(\cu_n)p) -   \nabla \hat \ell_p \right) _{y+\cu _{m}}\right|^{2}  \biggr)^{\!\sfrac12}. \notag\\
 &\qquad +C\sum_{m=n_0}^{n-1}3^{m-n} \biggl( 
\frac1{\left|\mathcal{Z}_{m}\right|}
\sum_{y\in \mathcal{Z}_{m}}
 \left| \left(  \nabla u(\cdot, y+ \cu_n, \ahom_*(\cu_n)p) -   \nabla \hat \ell_p \right) _{y+\cu _{m}}\right|^{2}  \biggr)^{\!\sfrac12}. \notag\\
& \leq
 C |p| 3^{-n}+
C\sum_{m=0}^{n_0}3^{m-n} 3^{n/4}|p| + C\sum_{m=n_0}^{n-1}3^{m-n} |p| 3^{- \frac 12 m(d-2+ \kappa \frac{d+2}{4(d-1)})}
\\ & \qquad 
+ 
C\sum_{m=n_0}^{n-1}3^{m-n}  |\ahom_*(\cu_m) -\ahom_*|^\frac 12 |p| \\
&\leq
C\sum_{m=n_0}^{n-1}3^{m-n} |\ahom_*(\cu_m) -\ahom_*|^\frac 12 |p|  + C 3^{-\alpha n} |p|.
    \end{align*}
This completes the proof of~\eqref{e.H-1u}.  
We obtain~\eqref{e.corrasconv} immediately from~\eqref{e.conv.bars} and~\eqref{e.H-1u}. 
\end{proof}

We show next that the two limiting matrices~$\ahom$ and~$\ahom_*$ in~\eqref{e.conv.bars} are in fact equal. 

\begin{lemma}
We have that~$\ahom_* = \ahom$.
\end{lemma}
\begin{proof}
The inequality~$\ahom_* \leq \ahom$ is easy. To see this, we apply the Fenchel duality~\eqref{e.fenchel} and the representation of~$\nu(\cu_n,p)$ and~$\nu^*(\cu_n,p)$,  to deduce that, for every 
$n\in\N$ with $3^n > \S_F$ and~$p,q\in\Rd$,
\begin{equation} 
\frac 12 p \cdot \ahom(\cu_n) p+\frac 12 q \cdot \ahom_*^{-1}(\cu_n) q
\ge
pq- C3^{-\frac12 n(d-2+ \kappa \frac{4-d}{2(d-1)})} |p||q| 
\,.
\end{equation}
Taking $q= \ahom(\cu_n)p$ we conclude $(1+C3^{-\frac12 n(d-2+ \kappa \frac{4-d}{2(d-1)})}) \ahom(\cu_n) \ge \ahom_*(\cu_n)  $.  Sending $n\to \infty$ yields~$\ahom_* \le \ahom$.

\smallskip

We turn to the proof of the reverse inequality,~$\ahom_* \le \ahom$. Define the boundary layer
$$
\mathcal B_{n,m}: = \{z+\cu_m: z\in 3^m \Zd \cap \cu_n,  \dist(z, \partial \cu_n) = 3^m\}\,. 
$$
We will work with sufficiently large~$n$ for which~$3^n$ is larger than an appropriate random minimal scale.

 Let $N_0<\infty$ be a large enough constant to be fixed later.  Recall the definition of $\S_2$ in~\eqref{e.S2}.  For every $3^n > \S_2$, take $\eta\in C_c^\infty(\Rd)$ to be a smooth cutoff function such that $\eta =1 $ in $\cu_n \setminus \mathcal B_{n,n-N_0}$ and $\eta =0 $ on $\mathcal B_{n,n/4}$.  Define 
$$
w := \hat \eta u (\cdot, \cu_n,  \ahom_*(\cu_n)p)+ (1-\hat \eta) \hat \ell_p\, ,
$$
where we recall the hat operation of a function defined in~\eqref{e.hat}.  
Clearly,  $w \in  \hat \ell_p + H_0^1 (\cu_n)$.  
Testing the minimizer of $\nu(\cu_n, p) $ with $w$ then yields
\begin{align*}
\frac 12 p \cdot \ahom(\cu_n) p
&\le
  \frac{1}{|\cu_n|} \sum_{e\in \mathcal{E}(\cu_n)}
\frac12 \nabla w \cdot \a \nabla  w\,.
\end{align*}
Notice that 
$$
\nabla w = \nabla \hat \eta  (u (\cdot, \cu_n,  \ahom_*(\cu_n)p) - \hat \ell_p) +  \hat \eta \nabla u (\cdot, \cu_n,  \ahom_*(\cu_n)p) +  (1-\hat \eta) \nabla \hat \ell_p .
$$
Therefore by Young's inequality,  we have for any $\eps>0$,  there exists $C_\eps <\infty$,  such that 
\begin{align*}
(\nabla w )^2 &\leq 
(1+\eps) (\hat \eta )^2 (\nabla u (\cdot, \cu_n,  \ahom_*(\cu_n)p) )^2 + 
C_\eps \left(( \nabla \hat \eta)^2  (u (\cdot, \cu_n,  \ahom_*(\cu_n)p) - \hat \ell_p)^2  
+ (1-\hat \eta)^2 (\nabla \hat \ell_p)^2\right) \\
&\leq
(1+\eps)  (\nabla u (\cdot, \cu_n,  \ahom_*(\cu_n)p) )^2 + 
C_\eps \left(( \nabla \hat \eta)^2  (u (\cdot, \cu_n,  \ahom_*(\cu_n)p) - \hat \ell_p)^2  
+ (1-\hat \eta)^2 (\nabla \hat \ell_p)^2\right),
\end{align*}
where we used the deterministic bound $0\le \hat \eta \le 1$ to obtain the last inequality.  Thus 
\begin{align}
\frac 12 p \cdot \ahom(\cu_n) p
&\le 
(1+\eps)   \frac{1}{|\cu_n|} \sum_{e\in \mathcal{E}(\cu_n)}
\frac12 \nabla u (\cdot, \cu_n,  \ahom_*(\cu_n)p) \cdot \a \nabla  u (\cdot, \cu_n,  \ahom_*(\cu_n)p) \notag \\
& \qquad + 
C_\eps \frac{1}{|\cu_n|} \sum_{e\in \mathcal{E}(\cu_n)} ( \nabla \hat \eta(e))^2  \a(e) (u (\cdot, \cu_n,  \ahom_*(\cu_n)p) - \hat \ell_p)^2
\notag \\ & \qquad   
+ 
C_\eps \frac{1}{|\cu_n|} \sum_{e\in \mathcal{E}(\mathcal B_{n, n-N_0+1})}  \a(e)  (\nabla \hat \ell_p(e))^2
\notag \\ &
\leq
(1+\eps) \frac12 p \cdot \ahom_*(\cu_n) p+ C_\eps \frac{1}{|\cu_n|} \sum_{e\in \mathcal{E}(\cu_n)} ( \nabla \hat \eta(e))^2  \a(e)(u (\cdot, \cu_n,  \ahom_*(\cu_n)p) - \hat \ell_p)^2  
\notag  \\ & \qquad  
+ 
C_\eps \frac{1}{|\cu_n|} \sum_{e\in \mathcal{E}(\mathcal B_{n, n-N_0+1})}  \a(e) (\nabla \hat \ell_p(e))^2.
\label{e.ahomupper}
\end{align}
Here we used the fact that since $n >\S_2$,  by taking $n$ large enough such that $3^{n-N_0+1}>3^{n/4}$,  $ \hat \eta (x)= 1$ if $x\in \cu_n \setminus B_{n, n-N_0+1}  $ to obtain the first inequality,  and using the definition of $\nu^*$ to obtain the second.  

For the rest of the proof,  we take $n\in\N$ such that $3^n > \max\{\S, 3^{N_0}\S_1,\S_2\}$,  where we recall the defition of $\S$ and $\S_1$ in Definition \ref{d.good} and Lemma \ref{l.goodS}.   By H\"older's inequality, for any $\delta>0$ we have 
\begin{multline}
\label{e.Holder}
\frac{1}{|\cu_n|} \sum_{e\in \mathcal{E}(\cu_n)} ( \nabla \hat \eta(e))^2 \a(e)  (u (\cdot, \cu_n,  \ahom_*(\cu_n)p) - \hat \ell_p)^2 
\\ 
\leq
\| u (\cdot, \cu_n,  \ahom_*(\cu_n)p) - \hat \ell_p\|_{\underline L^{2+\delta}(\cu_n)}^{2}
\| \nabla \hat \eta  \a \nabla \hat \eta  \|_{\underline L^{\frac{2+\delta}{\delta}}(\cu_n)} \,.
\end{multline}
We also notice that 
\begin{equation}
\nabla \hat \eta (e)  \a (e)\nabla \hat \eta (e) \leq 
\left\{
\begin{aligned}
& C_1 3^{2N_0}3^{-2n},  \quad \text{ if }  e\in \Z^d \setminus  \mathcal C_{>},
\\
&C_2 3^{2N_0} 3^{-2n}\diam{\mathcal C_{>}(x)}^2, \quad \text{ if }  e=(x,y) \in \partial \mathcal C_{>}
\,.
\end{aligned}
\right.
\end{equation}
Therefore 
\begin{align*}
\| \nabla \hat \eta  \a \nabla \hat \eta  \|_{\underline L^{\frac{2+\delta}{\delta}}(\cu_n)}
\leq
C_13^{2N_0}3^{-2n} + C_2 3^{2N_0} 3^{-2n} \biggl( \frac{1}{|\cu_n|}  \sum_{x\in \cu_n\cap \partial \mathcal C_{>} }\diam{\mathcal C_{>}(x)}^{\frac{4+2\delta}{\delta}} \biggr)^{\frac \delta{2+\delta}}\,.
\end{align*}
Upper bound the sum over the boundary by the volume,  and using the definition of good cube,  we conclude there exists $C_{\delta,d} <\infty$,  such that 
\begin{align*}
\sum_{x\in \cu_n\cap \partial \mathcal C_{>} }\diam{\mathcal C_{>}(x)}^{\frac{4+2\delta}{\delta}} 
 \leq
 |\cu_n\cap \mathcal C_{>} | \diam(\mathcal C_{>}\cap \cu_n)^{\frac{4+2\delta}{\delta}}
 \le
 \diam(\mathcal C_{>}\cap \cu_n)^{d+\frac{4+2\delta}{\delta}}
\le
C_{\delta,d} |\cu_n|.
\end{align*}
Thus 
$$
\| \nabla \hat \eta  \a \nabla \hat \eta  \|_{\underline L^{\frac{2+\delta}{\delta}}(\cu_n)}
\le
(C_1 + C_2C_{\delta,d})3^{2N_0}3^{-2n}\,.
$$
To bound the other term in~\eqref{e.Holder}, we use~$L^p$ interpolation in the form
\begin{align}
\label{e.interpolate}
\lefteqn{
\| u (\cdot, \cu_n,  \ahom_*(\cu_n)p) -  \hat \ell_p\|_{\underline L^{2+\delta}(\cu_n)}
} \qquad & 
\notag \\ & 
\leq
C\|u (\cdot, \cu_n,  \ahom_*(\cu_n)p) -  \hat \ell_p\|_{\underline L^{2}(\cu_n)}^\frac 12 \| u (\cdot, \cu_n,  \ahom_*(\cu_n)p) -  \hat \ell_p\|_{\underline L^{q}(\cu_n)}^\frac 12\,,
\end{align}
where $q:= \frac {2(2+\delta)}{2-\delta} $. 
We apply Lemma~\ref{l.H-1} and Lemma~\ref{l.corrector} to deduce that 
\begin{multline*}
3^{-n}\|u (\cdot, \cu_n,  \ahom_*(\cu_n)p) -  \hat \ell_p - (u (\cdot, \cu_n,  \ahom_*(\cu_n)p) -  \hat \ell_p)_{\cu_n}\|_{\underline L^{2}(\cu_n)}
\\
\leq 
C 3^{-n} \left\|\nabla u(\cdot, \cu_n,  \ahom_*(\cu_n)p) - \nabla \hat \ell_p \right\|_{\underline{H}^{-1}(\cu_{n+1})}
\longrightarrow 0 
\quad \mbox{as} \ n \to \infty\,.
\end{multline*}
Applying the Poincar\'e inequality and~\eqref{e.graduave}, we obtain 
\begin{align*}
3^{-n}\| (u (\cdot, \cu_n,  \ahom_*(\cu_n)p) -  \hat \ell_p)_{\cu_n}\|_{\underline L^{2}(\cu_n)}
&
\le
C \| (\nabla u (\cdot, \cu_n,  \ahom_*(\cu_n)p) - \nabla  \hat \ell_p)_{\cu_n}\|_{\underline L^{2}(\cu_n)}
\\ & 
\le
C |p| 3^{- \frac 12 n(d-2+ \kappa \frac{d+2}{4(d-1)})}
\to 
0 \quad \mbox{as} \ n\to \infty\,.
\end{align*}
Therefore
\begin{equation*}
\lim_{n\to \infty} 3^{-n}\| u (\cdot, \cu_n,  \ahom_*(\cu_n)p) -  \hat \ell_p\|_{\underline L^{2}(\cu_n)}
= 
0\,.
\end{equation*}
To bound the $\underline L^q$ norm in the right side of~\eqref{e.interpolate}, we apply the Sobolev inequality which yields a constant $C_q <\infty$,  such that 
\begin{align*}
\| u (\cdot, \cu_n,  \ahom_*(\cu_n)p) - \hat \ell_p\|_{\underline L^{q}(\cu_n)}
\le
C_q 3^n \| \nabla u (\cdot, \cu_n,  \ahom_*(\cu_n)p) - \nabla \hat \ell_p\|_{\underline L^{q^*}(\cu_n)},
\end{align*}
where $\frac 1{q^*} = \frac 1q + \frac 1d$,  $q^* = \frac{2d(2+\delta)}{2d+4+(2-d)\delta}$.  Taking $\delta = \frac 1d$,  we have $q^* = \frac{4d+2}{2d+4+\frac{2-d}d}\in (1,2)$.   We then apply the H\"older inequality which yields
\begin{align*}
\| \nabla u (\cdot, \cu_n,  \ahom_*(\cu_n)p) - \nabla \hat \ell_p\|_{\underline L^{q^*}(\cu_n)}
\leq
\bigl\| \a^\frac 12 (\nabla u (\cdot, \cu_n,  \ahom_*(\cu_n)p) - \nabla \hat \ell_p) \bigr\|_{\underline L^{2}(\cu_n)}^{q^*} \bigl\| \a^{-q^*/2} \bigr\|_{\underline L^{p}(\cu_n)}^\frac 1p\,,
\end{align*}
where $p : = \frac{2d+4+\frac{2-d}d}{3+(d-2)/d}$. By the definition of the good cubes and Corollary~\ref{c.mgf}, which gives bounds on the moments of~$\a^{-1}$,  there exists $C(d)<\infty$ such that 
$$
\| \a^{-q^*/2}\|_{\underline L^{p}(\cu_n)}^\frac 1p
\le
2\,\bigl\langle \a^{-pq^*/2} \bigr \rangle_{\mu_\tau} ^\frac 1p
\le 
C\,.
$$
Combining the estimates above,  we conclude that 
\begin{multline*}
\frac{1}{|\cu_n|} \sum_{e\in \mathcal{E}(\cu_n)} ( \nabla \hat \eta(e))^2  \a(e) (u (\cdot, \cu_n,  \ahom_*(\cu_n)p) - \hat \ell_p)^2
\\  
\le
C_{d }3^{2N_0-n} \bigl\|\nabla u(\cdot, \cu_n,  \ahom_*(\cu_n)p) - \nabla \hat \ell_p \bigr\|_{\underline{H}^{-1}(\cu_{n+1})}\,,
\end{multline*}
which vanishes~$\P$--a.s.~as $n\to \infty$ by Lemma~\ref{l.corrector}. 

Finally,  we notice that by~\eqref{e.graddiff.2}, and controlling the boundary of $\mathcal C_>$ by its volume,   we have 
\begin{equation}
 \frac{1}{|\cu_n|} \sum_{e\in \mathcal{E}(\mathcal B_{n, n-N_0+1}) \cap \partial \mathcal C_>}
\a(e) (\nabla  \hat\ell_p (e))^2 \\
\leq 
C_d KM|p|^2 \frac{\diam ( \mathcal C_>\cap \mathcal B_{n,n-N_0+1})^{d+2}}{|\cu_n|} \,.
\end{equation}
Recall that $n>N_0+ \log_3 \S_1$,  where $\S_1$ is the minimal scale at which \eqref{e.d+1moment} holds. We may therefore tile $\mathcal B_{n,n-N_0+1}$ with cubes of size $\cu_{n-N_0}$ to conclude 
\begin{align*}
C_d KM|p|^2 \frac{\diam ( \mathcal C_>\cap \mathcal B_{n,n-N_0+1})^{d+2}}{|\cu_n|}  
\leq
C_d KM  |p|^2 3^{-N_0}\,.
\end{align*}
This implies 
\begin{equation}
\frac{1}{|\cu_n|} \sum_{e\in \mathcal{E}(\mathcal B_{n, n-N_0+1})}   \a(e) (\nabla \hat \ell_p(e))^2
\leq
(C_d+1) KM |p|^2 3^{-N_0+1}\,.
\end{equation}
Therefore by first sending $n \to \infty$ and then $N_0 \to \infty$ in~\eqref{e.ahomupper}  we obtain 
$$
\frac 12 p \cdot \ahom p 
\le
(1+\eps) \frac 12 p \cdot \ahom_* p \,.
$$
Sending $\eps \to 0$ yields $\ahom \le \ahom_*$ and completes the proof of the lemma.
\end{proof}

\subsubsection{Homogenization theorem}

In this section we state Theorem~\ref{t.homogenize},   which shows a degenerate elliptic PDE~\eqref{e.HSfullvol} homogenizes to a constant coefficient equation~\eqref{e.homogpde} on large-scales.  We will also explain the connection between the PDE~\eqref{e.HSfullvol} and the central limit theorem and derive Theorem~\ref{t.CLT}.

\begin{theorem}
\label{t.homogenize}
Let~$\f$ satisfies the conditions in Theorem~\ref{t.CLT}, and let~$\a$ be the diagonal matrix with~$\a(e,e) = e^{\tau_e}$, where~$\tau$ is sampled from the Gibbs measure~\eqref{e.mutau}. Let~$u_R$ and~$u$ denote, respectively, the solutions of the equations
\begin{equation}
\label{e.HSfullvol}
- \nabla \cdot \a \nabla u_R =  \nabla\cdot \f\Bigl(\frac{\cdot}{R}\Bigr)
\quad \mbox{in}  \ \Zd
\end{equation}
and 
\begin{equation}
\label{e.homogpde}
- \nabla \cdot \ahom \nabla \bar u = \nabla\cdot \fhom  \quad \mbox{in} \ \Rd
\,. 
\end{equation}
Let $\S,\S'$ be the minimal scales defined in Definition~\ref{d.good} and Corollary~\ref{c.lowerbd} respectively.  Let $n_0\in\N$ be such that $3^{n_0} \ge \max\{\S, \S'\}$. Then,  for every $R>3^{n_0}$ , we have
\begin{equation}
\label{e.energy}
\left|  ( \nabla u_R, \a \nabla u_R) 
- (\nabla \bar u, \ahom \nabla \bar u) \right|
=o_R(1),
\end{equation}
\end{theorem}
\begin{proof}
By integration by parts, we have $( \nabla u_R, \a \nabla u_R)  = -(\f,   \nabla u_R) $ and $ (\nabla \bar u, \ahom \nabla \bar u)  =- (\fhom,  \nabla \bar u)$.  
Denote the energy quantities associated with~\eqref{e.HSfullvol} and~\eqref{e.homogpde} by
\begin{align*}
\left\{ 
\begin{aligned}
&
\mathsf E_{R, \f} [u_R] = ( \nabla u_R, \a \nabla u_R) - (\f,   \nabla u_R) \, , \\ & 
\mathsf E_{\fhom} [\bar u] = (\nabla \bar u, \ahom \nabla \bar u)- (\fhom,  \nabla \bar u) \, .
\end{aligned}
\right. 
\end{align*}
Therefore~\eqref{e.energy} is equivalent to the convergence of $|R^{-d}\mathsf E_{R, \f} [u] - \mathsf E_{\fhom} [\bar u] |$.  In the proof we work with three scales $R\gg n\gg n_0$,  where $n$ is the scale at which the subadditive quantities and its minimizers almost converge,  and $R$ is the scale at which homogenization occurs.  

\emph{The lower bound.} To prove the lower bound,  take $[\log_3 R] > n> 2n_0$,  and tile $\Zd$ with the simplexes $S_n := \{ \triangle_n^\pi(z)  \,:\, n\in\Z\,,\, z\in 3^n\Zd\}$.   Notice that,  for every $z\in 3^n\Zd$ and $\pi \in \mathcal P$,  by testing the definition 
of $\nu^*$ with $u_R$,  we have 
\begin{equation*}
\nu^* ( \triangle_n^\pi(z) , q) \ge
\frac 1 {| \triangle_n^\pi(z)|}  \sum_{e\in \mathcal E( \triangle_n^\pi(z))} \Bigl( -\frac 12 \nabla u_R \cdot \a\nabla u_R + q\cdot \nabla u_R    \Bigr)\,.
\end {equation*}
In other words,
\begin{equation*}
\frac 1 {| \triangle_n^\pi(z)|}  \sum_{e\in \mathcal E( \triangle_n^\pi(z))} \Bigl( \frac 12 \nabla u_R \cdot \a\nabla u_R  + \f \cdot \nabla u_R  \Bigr)
\ge
\frac 1 {| \triangle_n^\pi(z)|}  \sum_{e\in \mathcal E( \triangle_n^\pi(z))}  (q+\f) \cdot \nabla u_R
- \nu^* ( \triangle_n^\pi(z) , q) \,.
\end {equation*}
Denote by $z_\pi$ the center of the simplex $ \triangle_n^\pi(z)$. By taking $ q= \ahom  \nabla \bar u(z_\pi) $,  and applying Lemma~\ref{l.simpconv} which yields 
\begin{align*}
 \nu^* ( \triangle_n^\pi(z) ,  \ahom  \nabla \bar u(z_\pi)) &= \bar \nu^* (  \ahom  \nabla \bar u(z_\pi) ) + o_n(1) |\nabla \bar u(z_\pi)|^2  = \frac 12  \nabla \bar u(z_\pi) \cdot \ahom  \nabla \bar u(z_\pi)+ o_n(1)|\nabla \bar u(z_\pi)|^2\,,
\end{align*}
 we obtain that 
\begin{align*}
\lefteqn{ \frac 1 {| \triangle_n^\pi(z)|}  \sum_{e\in \mathcal E( \triangle_n^\pi(z))} \left( \frac 12 \nabla u_R \cdot \a\nabla u_R  + \f \cdot \nabla u_R  \right) 
} \qquad & 
\notag \\ &
\ge
\frac 1 {| \triangle_n^\pi(z)|}  \sum_{e\in \mathcal E( \triangle_n^\pi(z))} \left( -\frac 12 \nabla \bar u (z_\pi)\cdot \ahom \nabla \bar u (z_\pi)+ ( \ahom \nabla \bar u (z_\pi) +\f) \cdot \nabla u_R   \right)+o_n(1)|\nabla \bar u(z_\pi)|^2 
\,.
 \end {align*}
Summing over the simplexes, we obtain 
\begin{align*}
\lefteqn{ R^{-d} \mathsf E_{R, \f} [u_R] } \qquad & 
\notag \\ & 
\ge
R^{-d}  \sum_{e\in \mathcal E( \Zd)} \left( -\frac 12 \nabla \bar u (z_\pi)\cdot \ahom \nabla \bar u (z_\pi)+ ( \ahom \nabla \bar u (z_\pi) +\f) \cdot \nabla u_R   \right)+o_n(1) \sum_{z\in 3^n \Zd} |\nabla \bar u(z_\pi)|^2 
\,.
\end {align*}
To estimate the cross term,  define the modulus of continuity 
\begin{align}
\label{e.omegaR}
\omega_R := \max_{x\in \triangle_n^\pi(z)} \left|  \f'(x)+ \nabla\nabla \bar u(x) \right|\,,
 \end {align}
 it is clear that for every $x\in \triangle_n^\pi(z)$,  $|\nabla\bar u(x) - \nabla\bar u (z_\pi) | \le \frac{3^n}R \omega_R$.  We may then write 
 \begin{align*}
& R^{-d}  \sum_{e\in \mathcal E( \Zd)} ( \ahom \nabla \bar u (z_\pi) +\f) \cdot \nabla u_R 
 \ge
  R^{-d}  \sum_{e\in \mathcal E( \Zd)} ( \ahom  \nabla \bar u+ \f)   \cdot \nabla u_R
  - \frac{3^n}R   R^{-d}  \sum_{e\in \mathcal E( \Zd)} \omega_R | \nabla u_R  |\,.
 \end {align*} 
 Using~\eqref{e.homogpde} we have for the lattice approximation 
  \begin{align*}
 | \sum_{e\in \mathcal E( \Zd)}   ( \ahom  \nabla \bar u+ \f)   \cdot \nabla u_R  |
 \le 
  \frac 1R \sum_{e\in \mathcal E( \Zd)} \omega_R(x) | \nabla u_R|\,.
   \end {align*} 
Therefore, by  Cauchy-Schwarz 
    \begin{align*}
     |R^{-d} \sum_{e\in \mathcal E( \Zd)}  ( \ahom  \nabla \bar u+ \f)   \cdot \nabla u_R  |
     &\le
      \Biggl( R^{-d}  \sum_{e\in \mathcal E( \Zd)} \a(e) | \nabla u_R(e)|^2 \Biggr)^\frac 12 
         \Biggl( R^{-d}  \sum_{e\in \mathcal E( \Zd)} \a^{-1}(e) \frac 1{R^2} | \omega_R(x)|^2 \Biggr)^\frac 12  \\
        & \le 
     \left( 2  R^{-d}  \mathsf E_{R,f}[u_R] \right)^\frac 12 
         \Biggl( R^{-d}  \sum_{e\in \mathcal E( \Zd)} \a^{-1}(e)\frac 1{R^2} | \omega_R(x)|^2 \Biggr)^\frac 12 .
           \end {align*} 
Denote by $\mathcal{E}(z+\cu_{n_0})$ the collection of the subcubes of size $3^{n_0}$ that  partitions $\Zd$,  by definition of the good cube,  
          \begin{align*}
          3^{-n_0 d}  \sum_{e\in \mathcal E( z+\cu_{n_0})} \a^{-1}(e) | \omega_R(x)|^2
          \le
          2 \langle \a^{-1}(e) \rangle_{\mu_\tau} \max_{x\in z+\cu_{n_0}} | \omega_R(x)|^2\,.
           \end {align*} 
Notice that for $|x|> 10R$ and any $k\in\N$,  $ |\nabla^k \bar u(x)| \le C|x|^{-(d-1+k)}$.  Therefore 
  \begin{align*}
\lefteqn{R^{-d}  \sum_{e\in \mathcal E( \Zd)} \a^{-1}(e) \frac 1{R^2}  | \omega_R(x)|^2
} \qquad & 
\notag \\   &     \leq
      R^{-d}  R^{-2}\sum_{z\in 3^{n_0}\Zd}   3^{n_0 d} 2 \langle \a^{-1}(e)  \rangle_{\mu_\tau} \max_{x\in z+\cu_{n_0}}  \omega_R(x)^2
      \\
   &    \leq
        R^{-d}R^{-2}\sum_{|z| < 10R}   C3^{n_0 d} 2 \langle \a^{-1}(e)  \rangle _{\mu_\tau}+  R^{-d}R^{-2}\sum_{|z| \ge 10R}   C3^{n_0 d} 2 \langle \a^{-1}(e)  \rangle_{\mu_\tau} \frac 1{|z|^{2d+2}}
  \leq
         \frac {C(n_0)}{R^2}.
        \end {align*} 
        This implies 
       $$
       |R^{-d} \sum_{e\in \mathcal E( \Zd)}  ( \ahom  \nabla \bar u+ \f)   \cdot \nabla u_R  |
       \le
        \frac {C(n_0)}R.
        $$ 
       Finally,  we notice that 
\begin{align*}
\lefteqn{ - \frac 1 {| \triangle_n^\pi(z)|}  \sum_{e\in \mathcal E( \triangle_n^\pi(z))} \left( \frac 12 \nabla \bar u (z_\pi)\cdot \ahom \nabla \bar u (z_\pi)  \right)
} \qquad &
\notag \\ &  
 \ge
        -   \fint_ {\triangle_n^\pi(z)} \frac 12 \nabla \bar u \cdot \ahom \nabla \bar u   
-
          \frac 1 {| \triangle_n^\pi(z)|}  \sum_{e\in \mathcal E( \triangle_n^\pi(z))} \left( \frac 12 \nabla(\bar u - \bar u (z_\pi))\cdot \ahom \nabla(\bar u - \bar u (z_\pi))  \right) \,.
           \end {align*} 
   The last term in the display above can be bounded by 
    \begin{align*}
    \frac C {| \triangle_n^\pi(z)|}  \sum_{e\in \mathcal E( \triangle_n^\pi(z))}  \frac{3^{2n}}{R^2}\omega_R^2 \,.
     \end {align*} 
Sum over all the simplexes,  and using the fact that $|\omega_R(z) | \le \frac C {|z|^d}$ for $|z|>2R$,  we have 
\begin{align*}
R^{-d}  \sum_{e\in \mathcal E(\Zd)}  \frac{3^{2n}}{R^2}\omega_R^2   \le
\frac {C3^{2n}}{R^2}\,.
 \end {align*} 
 Combining all the displays above and sending $R\gg n \to \infty$,  we conclude the lower bound  by an integration by parts
 \begin{align*}
 R^{-d}\mathsf E_{R, \f} [u_R]  
& \ge  \fint \ -\frac 12 \nabla \bar u \cdot \ahom \nabla \bar u  - o_R(1) = \fint  \Bigl(\frac 12 \nabla \bar u \cdot \ahom \nabla \bar u  + \f \cdot \nabla \bar u \Bigr) - o_R(1) .
  \end {align*} 
This completes the proof of the lower bound.     

\smallskip
   
\emph{The upper bound.}  We notice that $u_R$ is the minimizer for $\mathsf E_{R,\f} [\cdot]$. To find an upper bound,  for every $z\in 3^n\Zd$ and $\pi \in \mathcal P$,  Let $v^\pi_z$ be the minimizer
of $\nu(\triangle_n^\pi(z),  \nabla \bar u(z_\pi))$.  Define $\tilde u $ to be the gluing of $v^\pi(z)$: it is tempting to define
\begin{equation*} \label{}
\tilde u(x):= v^\pi_z(x,  \nabla \bar u(z_\pi)),\quad \text{if }x\in \triangle_n^\pi(z)\,.
\end{equation*}
However,  the function $v^\pi_z$ may not match along the boundary of the simplexes.  We will add perturbations to them so that it is well defined along the boundary.  Define 
\begin{equation}
\tilde \ell (x) := 
\left\{
\begin{aligned}
& \hat \ell_{\nabla \bar u(z_\pi)}(x),  \quad \text{ if }  x\in \triangle_n^\pi(z) \setminus  \mathcal C_{>},
\\
& \frac 1{|\partial \mathcal C_{>}(x)|} \sum_{y \in \partial \mathcal C_{>}(x)} g(y), \quad \text{ if }  x\in   \triangle_n^\pi(z)\cap \mathcal C_{>}
\,.
\end{aligned}
\right.
\end{equation}
where $g(y) = \nabla \bar u(z_\pi)$ if $y \in \triangle_n^\pi(z)$.   We then define 
\begin{equation*} \label{}
\tilde u(x):= v^\pi_z(x,   \nabla \bar u(z_\pi)) + \tilde \ell(x) -\hat \ell_{\nabla \bar u(z_\pi)}(x) ,\quad \text{if }x\in \triangle_n^\pi(z)\,.
\end{equation*}
We then have 
\begin{align*} 
\mathsf E_{R,\f} [u_R]
&\leq
\mathsf E_{R,\f} [\tilde u]
\\ &
= 
\sum_{z\in 3^n\Zd} \sum_{e\in \mathcal E( \triangle_n^\pi(z))} \left( \frac 12 \nabla \tilde u \cdot \a\nabla  \tilde u  + \f \cdot \nabla  \tilde u \right)  \\
&\leq
\sum_{z\in 3^n\Zd} \sum_{e\in \mathcal E( \triangle_n^\pi(z))} |\triangle_n^\pi(z)|  \left(\nu (\triangle_n^\pi(z),   \nabla \bar u(z_\pi)) +\f(z_\pi)   \nabla \bar u(z_\pi) \right)
\\ & \qquad 
+ \sum_{z\in 3^n\Zd} \sum_{e\in \mathcal E( \triangle_n^\pi(z))} |-\f(z_\pi) +  \f| \cdot \nabla  \tilde u \\
& \qquad + \sum_{z\in 3^n\Zd}\sum_{e\in \mathcal E( \triangle_n^\pi(z))} \frac 12 \nabla (\tilde \ell - \hat \ell_{\nabla \bar u(z_\pi)} ) \cdot \a\nabla (\tilde \ell - \hat \ell_{\nabla \bar u(z_\pi)} )
\\ & \qquad 
+\sum_{z\in 3^n\Zd} \sum_{e\in \mathcal E( \triangle_n^\pi(z))} \f(z_\pi) |   \nabla \bar u(z_\pi)-\nabla \tilde u |\,,
\end{align*} 
Applying Lemma~\ref{l.simpconv} yields 
$$ \nu(\triangle_n^\pi(z),  \nabla \bar u(z_\pi)) = \bar \nu(\nabla \bar u(z_\pi) ) + o_n(1)  |\nabla \bar u(z_\pi)|^2= \frac 12  \nabla \bar u(z_\pi)\cdot \ahom \nabla \bar u(z_\pi)  +o_n(1)|\nabla \bar u(z_\pi)|^2\,.$$
A similar argument as in the lower bound gives 
\begin{align*}
\lefteqn{
\Biggl| \frac 1 {| \triangle_n^\pi(z)|}  \sum_{e\in \mathcal E( \triangle_n^\pi(z))}   (-\f(z_\pi) +  \f) \cdot \nabla  v^\pi_z \Biggr|
} \qquad & 
\notag \\ & 
\leq
 \frac{3^n}R \frac 1 {| \triangle_n^\pi(z)|} \Biggl( \sum_{e\in \mathcal E( \triangle_n^\pi(z))} \a(e)  (\nabla v^\pi(z)(e) )^2\Biggr)^\frac 12 
  \Biggl( \sum_{e\in \mathcal E( \triangle_n^\pi(z))} \a^{-1}(e)  \Biggr)^\frac 12  \\
 & \leq
  C\frac{3^n}R\langle \a^{-1}(e) \rangle_{\mu_\tau} \nu(\triangle_n^\pi(z),  \nabla \bar u(z_\pi))
  \le
  C\frac{3^n}R\langle \a^{-1}(e) \rangle_{\mu_\tau} |\nabla \bar u(z_\pi)|^2\,.
  \end{align*} 
  Thus 
   \begin{align*}
  \Biggl| R^{-d}  \sum_{e\in \mathcal E(\Zd)}   (-\f(z_\pi) +  \f) \cdot \nabla  v^\pi_z\Biggr|
   \le
   C\frac{3^n}R\langle \a^{-1}(e) \rangle_{\mu_\tau}  R^{-d}3^n\sum_{z\in 3^n\Zd} |\nabla \bar u(z_\pi)|^2
   \le
   C'\frac{3^{2n}}R\langle \a^{-1}(e) \rangle_{\mu_\tau}\,.
   \end{align*} 
To estimate the perturbation term, we use that $\bar u $ is smooth, which yields, 
for $e=(x,y)\in \mathcal{E}' \cap \partial \mathcal C_>$,
\begin{align}
|\nabla \tilde\ell(e)- \nabla  \hat \ell_{\nabla \bar u(z_\pi)}(e)|
&
\leq 
\max_{|x-y| \le \diam ( \mathcal C_>(x)) + 2\diam|\cu_n|} |\nabla \bar u (x) -\nabla \bar u (y)| 
\notag \\ &
\le
C ( \diam ( \mathcal C_>(x))  + 3^n) R^{-1} \omega_{n,2} (x)\,,
\end{align}
where we apply the definition of good cubes Definition~\ref{d.good} to conclude $ \diam ( \mathcal C_>(x)) \le C_0 3^n$ for some $C_0<\infty$,   and  $$\omega_{n,2} (x): = \max_{|x-y| \le (C_0+  2)\diam|\cu_n|} |\nabla \nabla \bar u (y)|\,.$$
Using the boundary of a set can be controlled by its volume,  ~$|\partial \mathcal C_>| \leq | \mathcal C_>|\leq C_d \diam ( \mathcal C_>(x))^{d}$, and we deduce, for every $x\in \partial \mathcal C_>$,
\begin{align}
\lefteqn{
3^{-nd} \sum_{e\in \mathcal{E}( \triangle_n^\pi(z))) \cap \partial \mathcal C_>}
 \frac 12 \nabla (\tilde \ell - \hat \ell_{\nabla \bar u(z_\pi)} ) \cdot \a\nabla (\tilde \ell - \hat \ell_{\nabla \bar u(z_\pi)} ) 
 } \qquad & 
 \notag \\ &
\leq 
C ( \diam ( \mathcal C_>\cap \cu_n)  + 3^n)^{2}  \diam ( \mathcal C_>\cap \cu_n)^d  3^{-nd}R^{-2}\omega_{n,2}^2 (x)\,.
\end{align}
Using the definition of good cubes in Definition~\ref{d.good},   $ \diam ( \mathcal C_>\cap \cu_n) ^{d+2} \le C|\cu_n|$,  and therefore the above quantity is $O((3^n/R)^2\omega_{n,2}^2 (x))$.  
Using the summability of $\omega_{n,2}^2$,  we may therefore conclude 
\begin{align}
R^{-d} \sum_{e\in \mathcal E(\Zd)}  \frac 12 \nabla (\tilde \ell - \hat \ell_{\nabla \bar u(z_\pi)} ) \cdot \a\nabla (\tilde \ell - \hat \ell_{\nabla \bar u(z_\pi)} ) 
\leq
C (3^n/R)^2 R^{-d}\sum_{x\in\Zd}\omega_{n,2}^2 (x)
\leq
C (3^n/R)^2 .
\end{align}
Similarly,
\begin{align*}
&\biggl| R^{-d} \sum_{e\in \mathcal E(\Zd)}  \f(z_\pi) \cdot (\nabla \bar u(z_\pi) -  \nabla\tilde u )\biggr| \\
& \leq
\biggl| R^{-d} \sum_{e\in \mathcal E(\Zd)}  \f(z_\pi) \cdot (\nabla  \hat \ell_{\nabla \bar u(z_\pi)} -  \nabla  \tilde \ell_{\nabla \bar u(z_\pi)})\biggr|
+
\biggl| R^{-d} \sum_{e\in \mathcal E(\Zd)}  \f(z_\pi) \cdot (\nabla  \hat \ell_{\nabla \bar u(z_\pi)} -  \nabla \bar u(z_\pi)\biggr|
\notag \\ & 
\leq
\frac{C3^{n}}{R}\sum_{x\in\Zd}\omega_{n,2} (x)
\leq
\frac{C3^{n}}{R}
\,.
\end{align*} 
Therefore 
\begin{align}
\label{e.almostupp}
R^{-d}\mathsf E_{R,\f} [u_R]
&  \leq
R^{-d}  \sum_{z\in 3^n\Zd } | \triangle_n^\pi(z)| \left(\frac 12   \nabla \bar u(z_\pi)\cdot \ahom \nabla \bar u(z_\pi) + \f(z_\pi) \cdot \nabla \bar u(z_\pi)  \right)  +C\frac{3^{2n}}R\,.
\end{align} 
Notice that 
\begin{align*}
\lefteqn{
\fint_{| \triangle_n^\pi(z)|}  \left(\frac 12   \nabla \bar u(z_\pi)\cdot \ahom \nabla \bar u(z_\pi) +\f(z_\pi) \cdot \nabla \bar u(z_\pi)  \right) 
} \qquad & 
\notag \\ & 
\leq
\fint_{| \triangle_n^\pi(z)|}  \left(\frac 12   \nabla \bar u\cdot \ahom \nabla \bar u +\f \cdot \nabla \bar u  \right) \\ & \qquad 
+ 
\fint_{| \triangle_n^\pi(z)|}  \left(\frac 12   \nabla (\bar u(z_\pi) - \bar u)\cdot \ahom \nabla ( \bar u(z_\pi) -\bar u)+ \f(z_\pi) \cdot \nabla (\bar u(z_\pi)-\bar u)  \right) \,.
 \end{align*}      
Recall the definition of $\omega_R$ in~\eqref{e.omegaR}.  Similar to the proof of the lower bound, we may sum over all simplexes and using the fact that $|\omega_R(z) | \le \frac C {|z|^d}$ for $|z|>2R$,  to conclude 
\begin{align*}
 & \sum_{z\in 3^n\Zd } | \triangle_n^\pi(z)| \left(\frac 12   \nabla \bar u(z_\pi)\cdot \ahom \nabla \bar u(z_\pi) +\f(z_\pi) \cdot \nabla \bar u(z_\pi)  \right)  \\
 & \qquad \leq
   R^{-d}\mathsf E_{\f} [\bar u]+
R^{-d}  \sum_{e\in \mathcal E(\Zd)} C \frac{3^{2n}}{R^2} \omega_R^2 + \|
\f\|_{L^\infty} \frac{3^{n}}{R} \omega_R  \le
 R^{-d}\mathsf E_{\f} [\bar u]+  \frac {C3^n} R\,.
 \end {align*} 
Combine with \eqref{e.almostupp},  this  yields $R^{-d}\mathsf E_{R,\f} [u_R]\le R^{-d}\mathsf E_{\f} [\bar u]+  \frac {C3^{2n}} R +o_R(1)$.  Sending $R\gg n \to \infty$ completes the proof of the upper bound. 
\end{proof}

\subsubsection{Proof of the CLT}

We turn to the proof of Theorem~\ref{t.CLT}. 

\begin{proof}[{Proof of Theorem~\ref{t.CLT}}]

 Given $f,g: \Zd \to\R$,  define the inner product
$$
(f,g)_R:= R^{-d} \sum_{x\in\Zd} f(x)g(x)\,.
$$
Given $f,g: \Rd \to\R$,  define
$$
(f,g):= \int f(x)g(x) \,dx\,.
$$
We now prove the moments of $F_R$ converges to the moments of the limiting Gaussian random variable. Obviously the odd moments are all zero by symmetry.  To compute the even moments, we notice that since the $\phi$ marginal is Gaussian, integrating out the $\phi$ marginal and apply the Wick's theorem to compute the Gaussian moments, we have
\begin{equation*}
\var_\mu [F_R^{2k}] 
= \left\langle (\nabla\cdot \f_R, \phi)_R^{2k}  \right\rangle_{\mu_\tau}
= (2k-1)!! \left\langle (\nabla\cdot \f_R,  G_\a \nabla\cdot \f_R)_R^k \right\rangle_{\mu_\tau}\,.
\end {equation*}
Alternatively, the right side above can be written as 
\begin{equation*}
(2k-1)!! \left\langle (\f_R, \nabla u_R )_R^k \right\rangle_{\mu_\tau},
\end {equation*}
where $u_R$ is the solution of~\eqref{e.HSfullvol}. The $2k$th moment of the limiting Gaussian variable can be written as
\begin{equation*}
(2k-1)!! \int_{\Rd} \nabla \cdot \f (x) \cdot G_{\ahom}(x,y)\nabla \cdot \f (y)\,dxdy 
= 
(2k-1)!!  (\f, \nabla \bar u )^k,
\end {equation*}
where $\bar u$ is the solution of~\eqref{e.homogpde}. 

\smallskip

For the second moment,  by Theorem~\ref{t.homogenize}, as $R\to\infty$
\begin{align*}
 \left|(\f_R, \nabla u_R )_R - (\f, \nabla \bar u ) \right| 
= 
\left|  ( \nabla u_R, \a \nabla u_R)_R 
- (\nabla \bar u, \ahom \nabla \bar u) \right|
\to 0.
\end{align*}
For general even moments,  we use the Brascamp-Lieb inequality (Proposition~\ref{p.BL}) to obtain
\begin{align*} 
\langle (\f_R, \nabla u_R )_R^k  \rangle_{\mu_\tau}
&
=
\frac 1 {(2k-1)!!} \left\langle (\nabla\cdot \f_R, \phi)_R^{2k}  \right\rangle_\mu
\notag \\ &
\leq 
C_k \Biggl( R^{-d}\sum_{e=(j,k)\in \mathcal E(\Zd)} ((\Delta^{-1} \nabla\cdot \f_R)_j -  (\Delta^{-1} \nabla\cdot \f_R)_k)^2\Biggr)^k.
\end{align*}
Since $\nabla \Delta^{-1}\nabla $ is a bounded operator $L^2 \to L^2$,  we have 
\begin{equation*} 
\langle (\f_R, \nabla u_R )_R^k  \rangle _{\mu_\tau}
\leq 
C'_k  (\f_R,\f_R)_R^k 
<\infty\,.
\end{equation*}
Now fix $k\in\N$,  by interpolating between the $L^1$ convergence of $(\f_R, \nabla u_R )_R$ and the $L^{2k}$ bound,  we conclude that, as $R\to \infty$,
$$
\left\langle (\f_R, \nabla u_R )_R^k \right\rangle_{\mu_\tau}
\to 
(\f, \nabla \bar u )^k\,.
$$
This completes the proof. 
\end{proof}

\appendix
\section{Appendix}

We need the following two standard functional inequalities. For functions defined in subsets of the continuum~$\Rd$, these can be found for instance in~\cite{AKMbook}. The same arguments are easily adapted to functions on the lattice, see for instance~\cite{AD1}. 

\begin{proposition}[Multiscale Poincar\'e inequality]
\label{p.MP}Fix $m\in\N$ and denote $\mathcal{Z}_{n}=3^{n}\mathbb{Z}^{d}\cap \cu _{m}$.  
Then,  for every $u\in L^2(\cu_m)$,
\begin{equation*}
\left\| u  \right\|_{\underline{H}^{-1}( \cu_{m}) }
\leq
C3^m\left|\left(u\right)_{\cu_{m}}\right| 
+
C\sum_{n=0}^{m-1}3^{n}\left( 
\frac1{\left|\mathcal{Z}_{n}\right|}
\sum_{y\in \mathcal{Z}_{n}}
 \left| \left(  u \right) _{y+\cu _{n}}\right|^{2}  \right)^{\!\sfrac12}.
\end{equation*}
\end{proposition}

\begin{lemma}
\label{l.H-1}
There exists $C_d<\infty$ such that for any $n\in\N$ and $u\in H^1(\cu_n)$,  we have 
\begin{equation*}
\| u - (u)_{\cu_n} \|_{\underline L^2(\cu_n)} \le 
C_d \| \nabla u \|_{\underline H^{-1}(\cu_n)} \,.
\end{equation*}
\end{lemma}

{\small
\bibliographystyle{alpha}
\bibliography{SOS}

\begin{thebibliography}{ABKM19}

\bibitem[ABKM19]{AdBKM}
S.~Adams, S.~Buchholz, R.~Koteck{\`y}, and S.~M{\"u}ller.
\newblock Cauchy-{B}orn rule from microscopic models with non-convex
  potentials, 2019.
\newblock preprint, arXiv:1910.13564.

\bibitem[AD18]{AD1}
S.~Armstrong and P.~Dario.
\newblock Elliptic regularity and quantitative homogenization on percolation
  clusters.
\newblock {\em Comm. Pure Appl. Math.}, 71(9):1717--1849, 2018.

\bibitem[AD22]{AD2}
S.~Armstrong and P.~Dario.
\newblock Quantitative hydrodynamic limits of the {L}angevin dynamics for
  gradient interface models, 2022.
\newblock preprint, arXiv:2203.14926.

\bibitem[AK81]{AkKr}
M.~A. Akcoglu and U.~Krengel.
\newblock Ergodic theorems for superadditive processes.
\newblock {\em J. Reine Angew. Math.}, 323:53--67, 1981.

\bibitem[AK22]{AKbook}
S.~Armstrong and T.~Kuusi.
\newblock Elliptic homogenization from qualitative to quantitative, 2022.
\newblock preprint, arXiv:2210.06488.

\bibitem[AKM16]{AdKM}
S.~Adams, R.~Koteck\'y, and S.~M\"uller.
\newblock Strict convexity of the surface tension for non-convex potentials,
  2016.
\newblock preprint, arXiv:1606.09541.

\bibitem[AKM19]{AKMbook}
S.~Armstrong, T.~Kuusi, and J.-C. Mourrat.
\newblock {\em Quantitative stochastic homogenization and large-scale
  regularity}, volume 352 of {\em Grundlehren der Mathematischen
  Wissenschaften}.
\newblock Springer-Nature, 2019.

\bibitem[AS16]{AS}
S.~N. Armstrong and C.~K. Smart.
\newblock Quantitative stochastic homogenization of convex integral
  functionals.
\newblock {\em Ann. Sci. \'Ec. Norm. Sup\'er. (4)}, 49(2):423--481, 2016.

\bibitem[AW22]{AW}
S.~Armstrong and W.~Wu.
\newblock {$C^2$} regularity of the surface tension for the {$\nabla\phi$}
  interface model.
\newblock {\em Comm. Pure Appl. Math.}, 75(2):349--421, 2022.

\bibitem[BFL82]{BFL}
J.~Bricmont, J.-R. Fontaine, and J.~L. Lebowitz.
\newblock Surface tension, percolation, and roughening.
\newblock {\em J. Statist. Phys.}, 29(2):193--203, 1982.

\bibitem[BL76]{BL}
H.~J. Brascamp and E.~H. Lieb.
\newblock On extensions of the {B}runn-{M}inkowski and {P}r\'ekopa-{L}eindler
  theorems, including inequalities for log concave functions, and with an
  application to the diffusion equation.
\newblock {\em J. Functional Analysis}, 22(4):366--389, 1976.

\bibitem[BLL75]{BLL}
H.~J. Brascamp, E.~H. Lieb, and J.~L. Lebowitz.
\newblock The statistical mechanics of anharmonic lattices.
\newblock In {\em Statistical Mechanics}, pages 379--390. Springer, 1975.

\bibitem[BPR22a]{Bau1}
R.~Bauerschmidt, J.~Park, and P.-F. Rodriguez.
\newblock The discrete {G}aussian model, {I}. renormalisation group flow at
  high temperature, 2022.
\newblock preprint, arXiv:2202.02286.

\bibitem[BPR22b]{Bau2}
R.~Bauerschmidt, J.~Park, and P.-F. Rodriguez.
\newblock The discrete {G}aussian model, {II}. infinite-volume scaling limit at
  high temperature, 2022.
\newblock preprint, arXiv:2202.02287.

\bibitem[BR18]{BR}
M.~Biskup and P.-F. Rodriguez.
\newblock Limit theory for random walks in degenerate time-dependent random
  environments.
\newblock {\em J. Funct. Anal.}, 274(4):985--1046, 2018.

\bibitem[BS11]{BiSp}
M.~Biskup and H.~Spohn.
\newblock Scaling limit for a class of gradient fields with nonconvex
  potentials.
\newblock {\em Ann. Probab.}, 39(1):224--251, 2011.

\bibitem[BS12]{BS}
D.~Brydges and T.~Spencer.
\newblock Fluctuation estimates for sub-quadratic gradient field actions.
\newblock {\em J. Math. Phys.}, 53(9):095216, 5, 2012.

\bibitem[CDM09]{CDM}
C.~Cotar, J.-D. Deuschel, and S.~M\"{u}ller.
\newblock Strict convexity of the free energy for a class of non-convex
  gradient models.
\newblock {\em Comm. Math. Phys.}, 286(1):359--376, 2009.

\bibitem[Dar23]{Da}
P.~Dario.
\newblock Upper bounds on the fluctuations of a class of degenerate convex
  $\nabla \phi $-interface models, 2023.
\newblock preprint, arXiv:2302.00547.

\bibitem[FS81]{FrSp}
J.~Fr\"{o}hlich and T.~Spencer.
\newblock The {K}osterlitz-{T}houless transition in two-dimensional abelian
  spin systems and the {C}oulomb gas.
\newblock {\em Comm. Math. Phys.}, 81(4):527--602, 1981.

\bibitem[FS97]{FS}
T.~Funaki and H.~Spohn.
\newblock Motion by mean curvature from the {G}inzburg-{L}andau {$\nabla \phi$}
  interface model.
\newblock {\em Comm. Math. Phys.}, 185(1):1--36, 1997.

\bibitem[Fun05]{Funaki}
T.~Funaki.
\newblock Stochastic interface models.
\newblock In {\em Lectures on probability theory and statistics}, volume 1869
  of {\em Lecture Notes in Math.}, pages 103--274. Springer, Berlin, 2005.

\bibitem[GOS01]{GOS}
G.~Giacomin, S.~Olla, and H.~Spohn.
\newblock Equilibrium fluctuations for {$\nabla\phi$} interface model.
\newblock {\em Ann. Probab.}, 29(3):1138--1172, 2001.

\bibitem[HS94]{HS}
B.~Helffer and J.~Sj\"ostrand.
\newblock On the correlation for {K}ac-like models in the convex case.
\newblock {\em J. Statist. Phys.}, 74(1-2):349--409, 1994.

\bibitem[Kir47]{Kir}
G~Kirchhofl.
\newblock {U}ber die {A}uflosung der {G}leichungen.
\newblock {\em aut welche man bei der Untersuchung der linearen Verteilung
  Glavanischer Strome gefurht wird, Ann. Phys. Chem.}, 72:497--508, 1847.

\bibitem[Lam22]{Lam}
P.~Lammers.
\newblock Height function delocalisation on cubic planar graphs.
\newblock {\em Probab. Theory Related Fields}, 182(1-2):531--550, 2022.

\bibitem[LO23]{LO}
P.~Lammers and S.~Ott.
\newblock Delocalisation and absolute-value-{FKG} in the solid-on-solid model.
\newblock {\em Probability Theory and Related Fields}, pages 1--25, 2023.

\bibitem[Mil11]{Mi}
J.~Miller.
\newblock Fluctuations for the {G}inzburg-{L}andau {$\nabla\phi$} interface
  model on a bounded domain.
\newblock {\em Comm. Math. Phys.}, 308(3):591--639, 2011.

\bibitem[MP15]{MiP}
P.~Mi{\l}o\'{s} and R.~Peled.
\newblock Delocalization of two-dimensional random surfaces with hard-core
  constraints.
\newblock {\em Comm. Math. Phys.}, 340(1):1--46, 2015.

\bibitem[MP22]{MP}
A.~Magazinov and R.~Peled.
\newblock Concentration inequalities for log-concave distributions with
  applications to random surface fluctuations.
\newblock {\em Ann. Probab.}, 50(2):735--770, 2022.

\bibitem[NS97]{NS}
A.~Naddaf and T.~Spencer.
\newblock On homogenization and scaling limit of some gradient perturbations of
  a massless free field.
\newblock {\em Comm. Math. Phys.}, 183(1):55--84, 1997.

\bibitem[She05]{She}
S.~Sheffield.
\newblock Random surfaces.
\newblock {\em Ast\'{e}risque}, (304):vi+175, 2005.

\bibitem[vB75]{vBe}
H.~van Beijeren.
\newblock Interface sharpness in the {I}sing system.
\newblock {\em Comm. Math. Phys.}, 40(1):1--6, 1975.

\bibitem[Wu22]{Wu}
W.~Wu.
\newblock Local central limit theorem for gradient field models, 2022.
\newblock preprint, arXiv:2202.13578.

\bibitem[Ye19]{Ye}
Zichun Ye.
\newblock Models of gradient type with sub-quadratic actions.
\newblock {\em J. Math. Phys.}, 60(7):073304, 26, 2019.

\end{thebibliography}
}

\end{document}